%% file: main.tex
\documentclass[11pt,]{article} 

\usepackage[left=1.25in,top=1.0in,right=1.25in,bottom=1.0in]{geometry}
\usepackage{microtype}
\usepackage{graphicx}
\usepackage{subcaption}
\usepackage{booktabs} 
\usepackage{amsmath}%
\usepackage{amssymb}
\usepackage{amsfonts}
\usepackage[round]{natbib}
\usepackage{hyperref}
\usepackage[capitalize]{cleveref}
\usepackage{multirow}
\usepackage{mathtools}
\usepackage{ marvosym }
\usepackage{tikz}
\usepackage{amsthm}
\usetikzlibrary{bayesnet}
\usepackage{setspace}
\usepackage{textcomp}
\usepackage{authblk} 
\usepackage{bbm}
\usepackage{array}
\usepackage[plain,noend,ruled,vlined]{algorithm2e}
\usepackage{rotating}
\usepackage[shortlabels]{enumitem}

\usepackage{thmtools,thm-restate}

\newtheorem{theorem}{Theorem}
\newtheorem{lemma}{Lemma}
\newtheorem{definition}{Definition}

\newtheorem{proposition}{Proposition}
\newtheorem{remark}{Remark}

\crefname{assumption}{Assumption}{Assumptions}
\Crefname{assumption}{Assumption}{Assumptions}

\let\oldReturn\Return
\renewcommand{\Return}{\State\oldReturn}

\newcommand{\R}{\mathbb{R}}

\newcommand{\rejected}{\hat{\mathcal{S}}}

\newcommand{\nulls}{\bar{\mathcal{S}}}

\newcommand{\DAG}{DAG}
\newcommand{\fdr}{FDR}

\newcommand{\fwer}{FWER}

\newcommand{\fdx}{FDX}

\newcommand{\defFDR}{ (FDR)}
\newcommand{\defFWER}{ (FWER)}
\newcommand{\defFDX}{ (FDX)}

\newcommand{\DAGGER}{DAGGER}

\newcommand{\BH}{BH}

\newcommand{\PRDS}{PRDS}
\newcommand{\PRDSadj}{PRDS}
\newcommand{\fullPRDS}{PRDS}
\newcommand{\fullPRDSadj}{PRDS}
\newcommand{\iid}{$\overset{\mathrm{i.i.d.}}{\sim}$}
\newcommand{\psmooth}{\tilde{p}}
\newcommand{\ust}{\uparrow^{\mathrm{st}}}
\newcommand{\E}{\mathbb{E}}
\newcommand{\lb}{\left(}
\newcommand{\rb}{\right)}
\renewcommand{\P}{\mathbb{P}}
\newcommand{\td}{\tilde}

\newcommand{\Z}{\mathbb{Z}}

\def\switchforrestateable#1{%
  \def\magic{\csname #1\endcsname} \magic*
}

\title{Smoothed Nested Testing on Directed Acyclic Graphs}
\date{}

\author[1,2]{\small Jackson H.~Loper\footnote{Equal contribution.}}
\newcommand\CoAuthorMark{\footnotemark[\arabic{footnote}]}
\author[3]{\small Lihua Lei\protect\CoAuthorMark}
\author[4]{\small William Fithian}
\author[5]{\small Wesley Tansey\thanks{\texttt{tanseyw@mskcc.org} (corresponding author)}}

\affil[1]{\footnotesize Data Science Institute, Columbia University, New York, NY, USA}
\affil[2]{\footnotesize Department of Neuroscience, Columbia University, New York, NY, USA}
\affil[3]{\footnotesize Department of Statistics, Stanford University, Palo Alto, CA, USA}
\affil[4]{\footnotesize Department of Statistics, University of California, Berkeley, CA, USA}
\affil[5]{\footnotesize Department of Epidemiology and Biostatistics, Memorial Sloan Kettering Cancer Center, New York, NY, USA}

\begin{document}

\maketitle
\begin{spacing}{1}
\begin{abstract}
\input{abstract}
\end{abstract}

\newpage

\input{introduction}

\input{background}

\input{method}

\input{testing}

\input{results}

\input{discussion}

\section*{Acknowledgments}
The authors thank Lizhen Lin and Kyoungjae Lee for helpful conversations with the early versions of this work. We also thank Thomas Norman for pointing out the genetic interaction maps application area.

\end{spacing}


\begin{small}
\bibliographystyle{abbrvnat}
\bibliography{main}
\end{small}

\singlespacing
\appendix
\input{appendix}

\end{document}

%% file: abstract.tex
We consider the problem of multiple hypothesis testing when there is a logical nested structure to the hypotheses. When one hypothesis is nested inside another, the outer hypothesis must be false if the inner hypothesis is false. We model the nested structure as a directed acyclic graph, including chain and tree graphs as special cases. Each node in the graph is a hypothesis and rejecting a node requires also rejecting all of its ancestors. We propose a general framework for adjusting node-level test statistics using the known logical constraints. Within this framework, we study a smoothing procedure that combines each node with all of its descendants to form a more powerful statistic. We prove a broad class of smoothing strategies can be used with existing selection procedures to control the familywise error rate, false discovery exceedance rate, or false discovery rate, so long as the original test statistics are independent under the null. When the null statistics are not independent but are derived from positively-correlated normal observations, we prove control for all three error rates when the smoothing method is arithmetic averaging of the observations. Simulations and an application to a real biology dataset demonstrate that smoothing leads to substantial power gains.

%% file: introduction.tex
\section{Introduction}
\label{sec:introduction}
We consider a structured multiple testing problem with a large set of null hypotheses structured along a directed acyclic graph. Each null hypothesis corresponds to a node in the graph and a node contains a false hypothesis only if all ancestors are false. The inferential goal is to maximize power while preserving a target error rate on the entire graph and rejecting hypotheses in a manner that obeys the graph structure. We will focus on boosting power in existing structured testing procedures by using the graph to share statistical strength between the node-level test statistics.

The graph-structured testing problem is motivated by modern biological experiments that collect a large number of samples on which to simultaneously test hundreds or even thousands of hypotheses. In genetics, for example, biologists are building ``genetic interaction maps'' \citep{costanzo:etal:2019:gene-interaction-maps-review}. These large networks outline how different genes rely on each other to produce or prevent certain phenotypes such as cell growth and death. Recent advances such as CRISPR-Cas9 \citep{wang:etal:2014:crispr-knockout-screens} and Perturb-Seq \citep{dixit:etal:2016:perturb-seq} enable biologists to experimentally disable hundreds of genes, both in isolation and in subsets of two or even three genes at once \citep{kuzmin:etal:2018:yeast-trigenic-interactions}. Testing the thousands of candidate sets of genes for differences from a control population is a classic multiple hypothesis testing problem.

Unlike the classical multiple testing problem, there is a rich structure to genetic interaction experiments. The biologist wishes to understand the sets as they relate to individual genes and subsets. For example, two genes may not produce a decrease in cell survival rates if only one of them is knocked out. However, knocking both out simultaneously in the same cell may lead to a sudden drop in survival rate \citep{costanzo:etal:2019:gene-interaction-maps-review}. In these cases, the biologist would consider the two genes to be ``interacting'' and thus the pair would be known as a ``synthetic lethal'' combination.

Beyond discovering the exact combination that leads to lethality, the biologist would also now flag the individual genes as having the potential to contribute to synthetic lethality, even though the genes cannot do so on their own. This potential has scientific and medicinal importance. For the scientist, if a gene is known to have the potential to contribute to cell death, it may be worth investigating it in the context of other knocked out genes that have not yet been considered. In medicine, if a specific type of cancer is seen to have a mutation in the first gene, a drug may be developed that inhibits the second gene, thereby killing the tumor cells. Thus, it is important to learn not just the exact lethal combinations, but also the entire ontology of genetic effects. 


Modern biological experiments aim to test not just individual genes or gene sets, but all entries in this ontology. In this structured testing problem, lower-level hypotheses are nested within higher-order hypotheses: if a gene or set of genes is truly associated with a change in phenotype, it logically entails that the subsets are also associated. Here we consider this nested testing problem in the general case, where a directed acyclic graph encodes an ontology of logically nested hypotheses; we will return to the genetic interact map example in \cref{sec:results}.

The key insight in this paper is that knowing the graph structure should increase power in the testing procedure. If every node in the graph represents an independent hypothesis test, then nodes should be able to borrow statistical strength from their ancestors and descendants. The signal at a nonnull node may be too weak to detect on its own, but the strength of evidence when combined with the evidence from its nonnull children may be sufficient to reject the null hypothesis. This strength sharing can also flow in the opposite direction, with nonnull parent nodes boosting power to detect nonnull children. For instance, if the nonnull signal attenuates smoothly as a function of depth in the graph, it may be possible to learn this function. The test statistic for a node can then be adjusted using the estimated signal predicted by the function learned from the ancestors. Whether using the descendants, ancestors, or both, sharing statistical strength creates dependence between test statistics and therefore must be carried out thoughtfully so as to enable control of the target error rate.

This leads us to develop a smoothing approach that implements this sharing of statistical strength between connected test statistics in the directed acyclic graph while still controlling the target error rate. We focus on descendant smoothing as it requires less prior knowledge of the graph and is thus more broadly applicable. We prove that the descendant smoothing approach yields adjusted $p$-values that are compatible with three different selection algorithms from the literature on nested testing with directed acyclic graphs \citep{meijer:goeman:2015:dag-fwer,genovese:wasserman:2006:fdx,ramdas:etal:2019:dagger}. Together, these techniques enable us to smooth the $p$-values and control familywise (Type I) error, the false discovery rate, or the false exceedance rate. Simulated and real data experiments confirm that smoothing yields substantially higher power across a wide range of alternative distributions and graph structures.

%% file: background.tex
\section{Background}
\label{sec:background}
There is a wealth of recent work on structured and adaptive testing. We focus on the most relevant work and refer the reader to \citet{lynch:2014:dissertation} for a comprehensive review of testing with logically nested hypotheses. 

Preserving the logical nesting structure after selection is the domain of structured testing \citep{shaffer1995multiple}. Methods for structured testing can be categorized based on their assumptions about the structure of the graph and the type of target error rate. For familywise error rate control, \citet{rosenbaum2008testing} propose a generic test on a chain graph; \citet{meinshausen:2008:tree-fwer} propose a procedure for testing on trees in the context of variable selection for linear regression; \citet{goeman2008multiple} propose the focus-level method, blending Holm's procedure \citep{holm:1979:bonferroni-holm} and closed testing \citep{marcus1976closed}, for testing on general directed acyclic graphs; and \citet{meijer:goeman:2015:dag-fwer} propose a more flexible method for arbitrary directed acyclic graphs based on the sequential rejection principle which unifies the aforementioned tests \citep{goeman:solari:2010:sequentialrp}. For false discovery rate control, the Selective Seqstep \citep{barber2015controlling}, Adaptive Seqstep \citep{lei2016power}, and accumulation tests \citep{li:barber:2017:accumulation_tests} enforce the logical constraint on the rejection set for chain graphs, though they do not require the logical constraint to actually hold; \citet{yekutieli2008hierarchical} propose a recursive procedure for testing on trees which provably controls the false discovery rate up to a computable multiplicative factor; \citet{lynch2016procedures} adapt the generalized step-up procedure to handle trees, which is further extended by \citet{ramdas:etal:2019:dagger} to general directed acyclic graphs.


There is a nascent literature on adaptive testing methods that preserve nested hypothesis structure. \citet{lei:etal:2017:star} describe an interactive adaptive procedure that partially masks $p$-values, enabling the scientist to explore the data and unveil its structure, then use the masked bits to perform selection while controlling the false discovery rate at the target level. This interactive approach is able to preserve the nested hypothesis structure and take advantage of covariates, but comes at the cost of splitting the $p$-values, potentially costing power. Further, the method is only able to control the false discovery rate for independent $p$-values; we will consider a much broader class of error metrics as well as some dependent $p$-value scenarios. The application of descendant smoothing  was also studied in \citet{vovk2020combining} for familywise error control and in \citet{ramdas:etal:2019:p-filter} for false discovery rate control, though the latter requires a conservative correction to handle the dependence induced by aggregation. A suite of methods \citep{scott:etal:2015:fdr-regression,xia:etal:2017:neuralfdr,tansey:etal:2018:bb-fdr,lei:fithian:2018:adapt,li:barber:2019:sabha} enable machine learning models to leverage side information like covariates that learn a prior over the probability of coming from the alternative; however, these methods do not enforce the logical constraint on the rejection set. 

Other methods go beyond classical error metrics, defining and controlling a structured error metric. \citet{benjamini2014selective} propose a method to control the average false discovery proportion over selected groups for a two-level graph, which is further extended by \citet{bogomolov:etal:2017:simes-tree} to general graphs. The p-filter \citep{barber:ramdas:2017:p-filter} and the multilayer knockoff filter \citep{katsevich:sabatti:2019:multilayer-knockoffs} are able to control the group-level false discovery rate simultaneously for potentially-overlapping partitions of hypotheses. Unlike the methods described in the last paragraph, for which the internal node in the graph can encode an arbitrary hypothesis, these four works seek to handle a special hierarchy where each internal node encodes the intersection of a subset of hypotheses on the leaf nodes. Our proposed approach is fundamentally different from these methods since we allow internal nodes to encode non-intersection hypotheses and our goal is to control the overall target error rate.

Rather than competing with methods for structured testing, our smoothing procedures are complementary. As we will show in \cref{sec:testing}, the descendant smoothing procedure is compatible with controlling familywise error via \citet{meijer:goeman:2015:dag-fwer}, the false exceedance rate via an extension to the method of \citet{genovese:wasserman:2006:fdx}, and the false discovery rate via the method of \citet{ramdas:etal:2019:dagger}. In the case of the latter method, we explicitly show that the smoothed $p$-values are positive regression dependent each on the subset of nulls \citep{benjamini:yekutieli:2001:fdr-dependence}, resolving the issue of dependence under the null after smoothing and hence avoiding conservative corrections as suggested by \citet{ramdas:etal:2019:dagger}. The benefit of smoothing is not to enable a new structured hypothesis testing procedure, but to make existing principled methods such as these more powerful by leveraging the structure of the problem. 


%% file: method.tex
\section{Smoothing Nested Test Statistics}
\label{sec:method}
\subsection[Smoothed p-values]{Smoothed $p$-values}
\label{subsec:method:smoothed-p}
Let $\{H_1,\ldots, H_n\}$ be a large set of null hypotheses. For each $v\in\{1,\ldots, n\}$ we observe a random variable $p_v$; if the null hypothesis $H_v$ holds, we assume that $p_v$ is super-uniform, i.e. $\mathrm{pr}(p_v \leq c) \leq c$ for any $c\in [0, 1]$. In some cases, we will assume the null $p$-value is uniform in $[0, 1]$. 
We will use $\mathcal{V}=\{1,\ldots,n\}$ to index all of the null hypotheses, and let
\[
\bar{\mathcal{S}} = \{v\in \mathcal V:\ H_v\ \mathrm{is\ true}\} \, , \quad \mathcal{S} = \mathcal{V}\setminus\bar{\mathcal{S}} = \{v\in \mathcal V:\ H_v\ \mathrm{is\ false}\}
\]
denote the unknown set of null hypotheses which hold and do not hold, respectively. 

Our task is to estimate which hypotheses are false: to produce an estimator $\hat{\mathcal{S}}$ of the set $\mathcal{S}$ from the random variables $\{p_v\}_{v \in \mathcal{V}}$. To help estimate $\mathcal{S}$, we have access to a directed acyclic graph $\mathcal{G} = (\mathcal{V}, \mathcal{E})$ whose edges encode constraints on the hypotheses in the following way: if $H_v$ is false and $w$ is an ancestor of $v$ in $\mathcal{G}$, then $H_w$ must also be false.  We would like to use these logical constraints to our advantage in estimating $\mathcal{S}$.  

To do so, we will use the graph $\mathcal{G}$ to transform the $p$-values into a new set of values and then apply existing structured testing procedures to the transformed values.  We call the transformed values smoothed $p$-values (denoted $\psmooth$), because they will be formed by various kinds of averages of the original $p$-values.  In the most general sense, the transformed values are created by the following process.  First, we select an arbitrary collection of smoothing functions $f_v:\ \mathbb{R}^{|\mathcal{V}|}\rightarrow \mathbb{R}$ (specific examples are given below).  For each $v\in\mathcal{V}$ let $\mathcal{C}_v$ denote the union of $v$ with all of its descendants in the graph $\mathcal G$. The smoothed $\psmooth$-value for node $v$ is then given by
\begin{equation}
  \label{eq:psmooth_v}
  F_v(c;x_{\mathcal{V}\backslash\mathcal{C}_v}) \triangleq \mathrm{pr}(f_v( u_{\mathcal{C}_v},x_{\mathcal{V}\backslash\mathcal{C}_v}) \leq c),\quad \psmooth_v \triangleq F_v(f_v(p_{\mathcal{C}_v},p_{\mathcal{V}\backslash\mathcal{C}_v});p_{\mathcal{V}\backslash\mathcal{C}_v}),
\end{equation}
where $u_{\mathcal{C}_v} = \{ u_w\}_{w\in\mathcal{C}_v}$ \iid{} $\mathrm{Uniform}[0,1]$. For the tailored functions $f_v$ discussed in \cref{subsec:method:descendants}, $F_v$ has a closed form expression. In general, it can be computed to arbitrary accuracy for any $f_v$ using Monte Carlo simulations.  The resulting random variable $\psmooth_v$ is a valid $p$-value for the hypothesis $H_v$,
\begin{restatable}[$\psmooth$-values are super-uniform]{lemma}{smoothp} \label{lem:smoothp}
  Assume that the null $p$-values are mutually independent and independent of nonnull $p$-values.
  \begin{enumerate}[(a)]
  \item If the null $p$-values are \iid{} as $\mathrm{Uniform}[0, 1]$, $\psmooth_v$ is super-uniform for every $v\in\bar{\mathcal{S}}$.
  \item If the null $p$-values are super-uniform and $f_v$ is nondecreasing in $p_{\mathcal{C}_v}$ for any value of $p_{\mathcal{V}\backslash\mathcal{C}_v}$, then $\psmooth_v$ is super-uniform for every $v\in\bar{\mathcal{S}}$.
  \end{enumerate}
\end{restatable}
We defer the proof to the appendix.  
\cref{lem:smoothp} allows us to construct a hypothesis test with Type I error control using the smoothed $\psmooth$-values. The value at $\psmooth_v$ is a function of $p_v$ and all its ancestors and descendants. This enables $\psmooth$-values to borrow statistical strength from each other and, depending on the choice of $\{f_v\}$, can lead to more powerful hypothesis tests.

\subsection{Descendant smoothing via merging}
\label{subsec:method:descendants}
The optimal smoothing functions $\{f_v\}_{v \in \mathcal{V}}$ are application specific. They depend on the structure of the graph, alternative hypothesis, and prior knowledge about the experiments. We focus our investigation on smoothing functions that use only the descendants at each node. Methods using ancestor nodes are left for future work. For notational convenience, we will write $f_v(p_{\mathcal{C}_v})$ for $f_v(p_{\mathcal{C}_v}, p_{\mathcal{V}\backslash \mathcal{C}_v})$ for descendant smoothing hereafter.

Descendant smoothing functions combine the $p$-values of descendant nodes with the current node to obtain smoothed $p$-values.  Many different strategies for combining $p$-values have been proposed in the literature, and it is beyond the scope of this work to investigate how each of them might perform as a descendant smoothing function. We instead consider a general class of smoothing functions derived by merging,
\[
f_v(p_{\mathcal{C}_v}) = G_v\left(\sum_{c\in \mathcal{C}_v}H_{v, c}^{-1}(p_c)\right) \, .
\]
This merging strategy covers many well-known methods for merging (independent) $p$-values. For instance, a Stouffer smoothing strategy would merge $p$-values following the method of \citet{stouffer:etal:1949:merging}, converting the $p$-values to $z$-scores and adding them,
  \[
    f_v(p_{\mathcal{C}_v}) = \sum_{c\in \mathcal{C}_v}\Phi^{-1}(p_c)
  \]
where $\Phi^{-1}(\cdot)$ is the distribution function of a standard normal. If $v$ corresponds to a null hypothesis, the logical constraint implies that $\{p_c: c\in\mathcal{C}_v\}$ are independent and uniformly distributed. As a result, $F_v(\cdot)$ is the distribution of a mean-zero normal distribution with variance $|\mathcal{C}_v|$.

A Fisher smoothing strategy would merge using the method of \citet{fisher:stats-methods} by considering the product of the $p$-values,
  \[
    f_v(p_{\mathcal{C}_v}) = \sum_{c\in \mathcal{C}_v}2\log(p_c).
  \]
When $H_v$ is null, $-f_v(p_{\mathcal{C}_v})$ has a chi-square distribution with degree of freedom $2|\mathcal{C}_v|$.
Fisher's method tends to have high power in a wide range of scenarios \citep[e.g.][]{littell1971asymptotic}, though other methods will be more powerful for certain alternative distributions. Other popular methods include Tippett's method, which takes the minimum $p$-values \citep{tippett1931methods, bonferroni1936teoria}; R\"{u}ger's method, which is based on an order statistic \citep{ruger1978maximale}; Simes' method \citep{simes1986improved} or higher criticism method \citep{donoho2004higher}, which combine all order statistics; the Cauchy combination, which aggregates inverse-Cauchy transformed $p$-values \citep{liu2020cauchy}; and the generalized mean aggregation method which aggregates monomial-transformed $p$-values \citep{vovk2020combining, vesely2021permutation}. \citet{heard:rubin-delanchy:2018:p-merging} provide a Neyman-Pearson analysis of optimal alternative hypotheses for each smoothing function; see also \cite{vovk2020admissible} for an admissibility analysis of different $p$-value aggregation methods under general dependence.

This class of descendant smoothing functions are convenient to work with both computationally and mathematically. Many have a closed form distribution that enables fast calculation of the smoothed statistic. As we will see, theoretical properties for a large class of smoothing methods can also be proven, making them compatible with a broad set of selection methods.




%% file: testing.tex
\section{Testing with Smoothed Statistics}
\label{sec:testing}

\subsection{Familywise error rate control}
\label{subsec:testing:fwer}

The familywise error rate\defFWER{}, controlled at level $\alpha$, ensures the probability that even one null hypothesis was rejected is at most $\alpha$, i.e. $\mathrm{pr}(\mid \rejected \cap \nulls \mid \geq 1) \leq \alpha$. It is a stringent error metric for multiple testing and useful in high-stakes decision-making where false positives are prohibitive.

To estimate $\mathcal{S}$ while controlling the \fwer{}, we can directly apply the algorithm of Meijer and Goeman to the $\psmooth$-values (\cite{meijer:goeman:2015:dag-fwer}), outlined in Appendix \ref{subapp:fwer}. The procedure provably controls the familywise error rate so long as the null p-values  are marginally all (super-)uniform. Therefore, by Lemma \ref{lem:smoothp}, even though the smoothed p-values are dependent, this correction is still valid.

The Lemma \ref{lem:smoothp} result is very general for controlling the \fwer{}. It admits any choice of function over the descendant and ancestor $p$-values. This is possible because the inner-loop of the \citet{meijer:goeman:2015:dag-fwer} algorithm relies on a Bonferroni correction. The union bound strategy of the Bonferroni correction places no requirement on the dependency structure of the statistics.

\subsection{False exceedance rate control}
\label{subsec:testing:fdx}

The false exceedance rate\defFDX{}, controlled at level $(\gamma, \alpha)$, ensures the false discovery proportion is greater than $\gamma$ with probability no greater than $\alpha$, i.e. $\mathrm{pr}(|\rejected \cap \nulls| / |\rejected| > \gamma) \leq \alpha$. It is less stringent than the familywise error rate. 

\cite{genovese:wasserman:2006:fdx} propose a generic procedure that turns a familywise error control method into a false exceedance rate control method. Specifically, starting from any rejection set $\rejected_0$ that controls the familywise error rate at level $\alpha$, we can append any subset $\mathcal{S'}\subset \mathcal{V}\setminus \rejected_0$ onto $\rejected_0$. Then the expanded rejection set $\rejected_0 \cup \rejected'$ controls the false exceedance rate if $|\rejected'| \le |\rejected_0| \gamma / (1 - \gamma)$ \citep[][Theorem 1]{genovese:wasserman:2006:fdx}. The proof is straightforward: on the event that $\rejected_0$ contains no false discovery, which has probability at least $1 - \alpha$, the false discovery proportion is at most $|\rejected'| / (|\rejected'| + |\rejected_0|) \le \gamma$.

To guarantee that $\rejected$ satisfies the logical constraint, we apply \cite{meijer:goeman:2015:dag-fwer}'s method to obtain $\rejected_0$, and then append another subset which does not violate the constraint. Since there is no restriction on $\rejected'$, we can greedily add hypotheses based on the topological ordering to maintain the constraint. The procedure is outlined in Appendix \ref{subapp:fdx}. 


\subsection{False discovery rate control}
\label{subsec:testing:fdr}

The false discovery rate\defFDR{}, controlled at level $\alpha$, ensures the expected proportion of rejected hypotheses that are actually null is at most $\alpha$, i.e. $\mathbb{E}[|\rejected \cap \nulls| / (1 \vee |\rejected|)] \leq \alpha$. This is one of the most popular error metrics in large scale inference.

To estimate $\mathcal{S}$ while controlling the \fdr{}, we apply the Greedily Evolving Rejections on Directed Acyclic Graphs method proposed by \cite{ramdas:etal:2019:dagger}. The method works with the original test statistics by extending the \citet{benjamini:hochberg:1995:bh} procedure to \DAG{}s. As with \BH{}, it is only guaranteed to control the \fdr{} if the $\psmooth$-values satisfy a special property: \fullPRDS{}.  Specifically, for any $x,y\in\mathbb{R}^m$, let $x\preceq y$ signify that $x_i \leq y_i$ for each $i$.  A set $D\in \mathbb{R}^m$ is called non-decreasing if $x \preceq y,x\in D \implies y\in D$.  A random object $X\in\mathbb{R}^m$ is said to satisfy \PRDS{} on $T\subset\{1,\cdots,m\}$ if $t\mapsto \mathbb{P}(X\in D|X_i=t)$ is non-decreasing for every non-decreasing set $D$ and every index $i\in T$.  

Various multiple testing procedures have been proven to control the \fdr{} under \PRDS{}. However, only a few concrete examples have been shown to satisfy this condition, such as the one-sided testing problem with nonnegatively correlated Gaussian statistics and the two-sided testing problem with t-statistics derived from uncorrelated z-values \citep{benjamini:yekutieli:2001:fdr-dependence}. This limits the practical usefulness of the theoretical guarantees established under this condition.

In Appendix \ref{sec:PRDS}, we establish a general theory of \fullPRDS{} based on classical stochastic ordering theory, a widely studied area in reliability theory \citep[e.g.][]{efron1965increasing, kamae1977stochastic, block1987probability}. In \cref{subsec:method:descendants}, we introduced a class of $p$-values derived from descendant smoothing techniques. Theorem \ref{thm:combination_test} shows that a broad class of these smoothed $p$-values satisfy \fullPRDS{}. 

\begin{theorem}\label{thm:combination_test}
  Assume that the null $p$-values are uniformly distributed in $[0, 1]$, mutually independent and independent of all nonnull $p$-values. For each node $v$ and its descendant $c\in \mathcal{C}_v$, let $H_{v, c}(x): \R\mapsto [0, 1]$ be a monotone increasing function with the first-order derivative $H_{v, c}'(x)$ being log-concave, and $G_v(x): \R\mapsto \R$ be a monotone increasing function. Further let
  \[f_v(p_{\mathcal{C}_v}) = G_v\left(\sum_{c\in \mathcal{C}_v}H_{v, c}^{-1}(p_c)\right).\]
  Then the smoothed $p$-values are \fullPRDS{}.
\end{theorem}

Theorem \ref{thm:combination_test} covers a broad class of descendant smoothing functions. For Stouffer smoothing, $G_v(x) = 1$, $H_{v, c}(x) = \Phi(x)$ and $H_{v, c}'(x) = \exp\{-x^2 / 2\} / \sqrt{2\pi}$, which is log-concave; for Fisher smoothing, $G_v(x) = 1$, $H_{v, c}(x) = \exp\{x / 2\}$ and $H_{v, c}'(x) = \exp\{x / 2\} / 2$, which is log-concave; for generalized mean smoothing \citep{vovk2020combining}, $G_v(x) = (x / |\mathcal{C}_v|)^{1/r}$, $H_{v, c}(x) = x^{1/r}$ and $H_{v, c}'(x) = x^{1/r - 1} / r$, which is log-concave if $0 < r \le 1$. For all these smoothing methods, Theorem \ref{thm:combination_test} covers their weighted versions with $H_{v, c}(x)$ replaced by $a_{v, c}H_{v, c}(x)$ for arbitrary $a_{v, c}\ge 0$, since the log-concavity of the derivative continues to hold.

Theorem \ref{thm:order_test} presents another class of smoothed $p$-values based on order statistics. It includes Tippett's method with $G_{v}(x) = |\mathcal{C}_v|x$ and $k_v = 1$ and R\"{u}ger's method with $G_{v}(x) = |\mathcal{C}_v|x / k$ and $k_v = k$. 

\begin{theorem}\label{thm:order_test}
  Under the same assumptions as in Theorem \ref{thm:combination_test}, the smoothed $p$-values satisfy \fullPRDS{} if
  \[f_v(p_{\mathcal{C}_v}) = G_v(p_{\mathcal{C}_v, (k_v)})\]
  where $p_{\mathcal{C}_v, (1)}\le p_{\mathcal{C}_v, (2)} \le \ldots \le p_{\mathcal{C}_v, (|\mathcal{C}_v|)}$ denote the order statistics of $p_{\mathcal{C}_v}$ and $k_v\in [1, |\mathcal{C}_v|] $ is an arbitrary integer.
\end{theorem}

Taken together, Theorems \ref{thm:combination_test} and \ref{thm:order_test} establish that \DAGGER{} will control the \fdr{} at the nominal level for most smoothing methods outlined in \cref{subsec:method:descendants}. 

\subsection{Dependent null statistics with Gaussian copulas}
\label{subsec:testing:dependence}
So far we have assumed that $p_{\bar{\mathcal{S}}}$ are independent and (super-)uniform. 
If this does not hold, the smoothed $\psmooth$ values are not guaranteed to be super-uniform.  This limits our ability to use these $\psmooth$-values for hypothesis testing. However, if anything is known about the dependency structure of the $p$-values then it may be possible to use this knowledge to create conservative bounds yielding super-uniform $\psmooth$-values. For instance, when the dependency structure between $p$-values is known, Fisher smoothing can be made valid by adjusting the critical value \citep{brown1975400, kost2002combining}.
However, when the correlation structure is unknown, Fisher smoothing is not recommended as it will likely lead to inflated false discovery rates. 

In this section we consider the case that $p_{\bar{\mathcal{S}}}$ carries a Gaussian copula with unknown correlation matrix $R$.  In other words, letting $\Phi^{-1}$ denote the quantile function of the standard normal and letting $Z_v = \Phi^{-1}(p_v)$, we will assume that $Z_{\bar{\mathcal{S}}} \sim \mathcal{N}(0,R)$.  This case arises naturally if the $p$-values come from correlated $z$-scores.  

When the copula is Gaussian, we can still control any of the three target error metrics by using a method we dub \emph{conservative Stouffer smoothing},
\[
\psmooth_v \gets 
    \begin{cases}
        1 & \mathrm{if } \sum_{w \in \mathcal{C}_v} \pi_{vw} Z_w\geq 0 \\
        \Phi\left(\sum_{w \in \mathcal{C}_v} \pi_{vw} Z_w\right) & \mathrm{otherwise}.
    \end{cases}
\]
where $\Phi$ is the cumulative distribution function of the standard normal and $\pi$ satisfies  $\pi_{vw}\geq 0,\sum_{w \in \mathcal{C}_v} \pi_{vw}=1$.

For \fwer{} and \fdx{}, it suffices to show that the smoothed $\psmooth$-values are marginally super-uniform.
\begin{restatable}[marginal validity for Gaussian copulas]{lemma}{mvnfwer}
\label{lem:mvnfwer}  Conservative Stouffer smoothing on marginally uniform $p$-values with any Gaussian copula yields super-uniform smoothed $\psmooth$-values, i.e. $pr(\psmooth_v \leq \alpha) \leq \alpha$ for all $v\in\bar{\mathcal{S}}$.
\end{restatable}
If we know that the nulls are all non-negatively correlated, 
we can prove the following result, implying that \DAGGER{} can control the \fdr{}.
\begin{restatable}[\PRDS{} for conservative Stouffer smoothing]{lemma}{mvnfdr}
\label{lem:mvnfdr}
Let $R$ be a correlation matrix with no negative entries.  Conservative Stouffer smoothing on marginally uniform $p$-values with any Gaussian copula of correlation $R$ yields smoothed $\psmooth$-values which are \fullPRDSadj{} on $\bar{\mathcal{S}}$.
\end{restatable}
Appendix \ref{subapp:depexperiments} shows several examples where this smoothing method yields improved power.

%% file: results.tex
\section{Results}
\label{sec:results}
\subsection{Simulations}
\input{simulations}
For each simulation, we run 100 independent trials and report empirical estimates of power, \fwer{}, \fdx{}, and \fdr{} at
\[
\alpha=(0.01, 0.02, 0.03, 0.04, 0.05, 0.08, 0.1, 0.15, 0.2, 0.25).
\]
For the \fdx{}, we fix $\gamma=0.1$ for all experiments. We compare performance with and without Fisher smoothing using the method of \citet{meijer:goeman:2015:dag-fwer} for \fwer{} control, \citet{meijer:goeman:2015:dag-fwer} with Algorithm~\ref{alg:fdx} for \fdx{} control, and \citet{ramdas:etal:2019:dagger} for \fdr{} control. Both structured testing methods are the current state of the art for testing on \DAG{}s, with both showing the highest power to-date relative to other methods targeting the same error rate. We also compare to the structureless method of \citet{benjamini:hochberg:1995:bh}, though this method does not preserve nesting structure. However, performance relative to this baseline illustrates how smoothing turns the graph structure into an advantage rather than just a constraint.

\begin{figure}[t!]
\includegraphics[width=\linewidth]{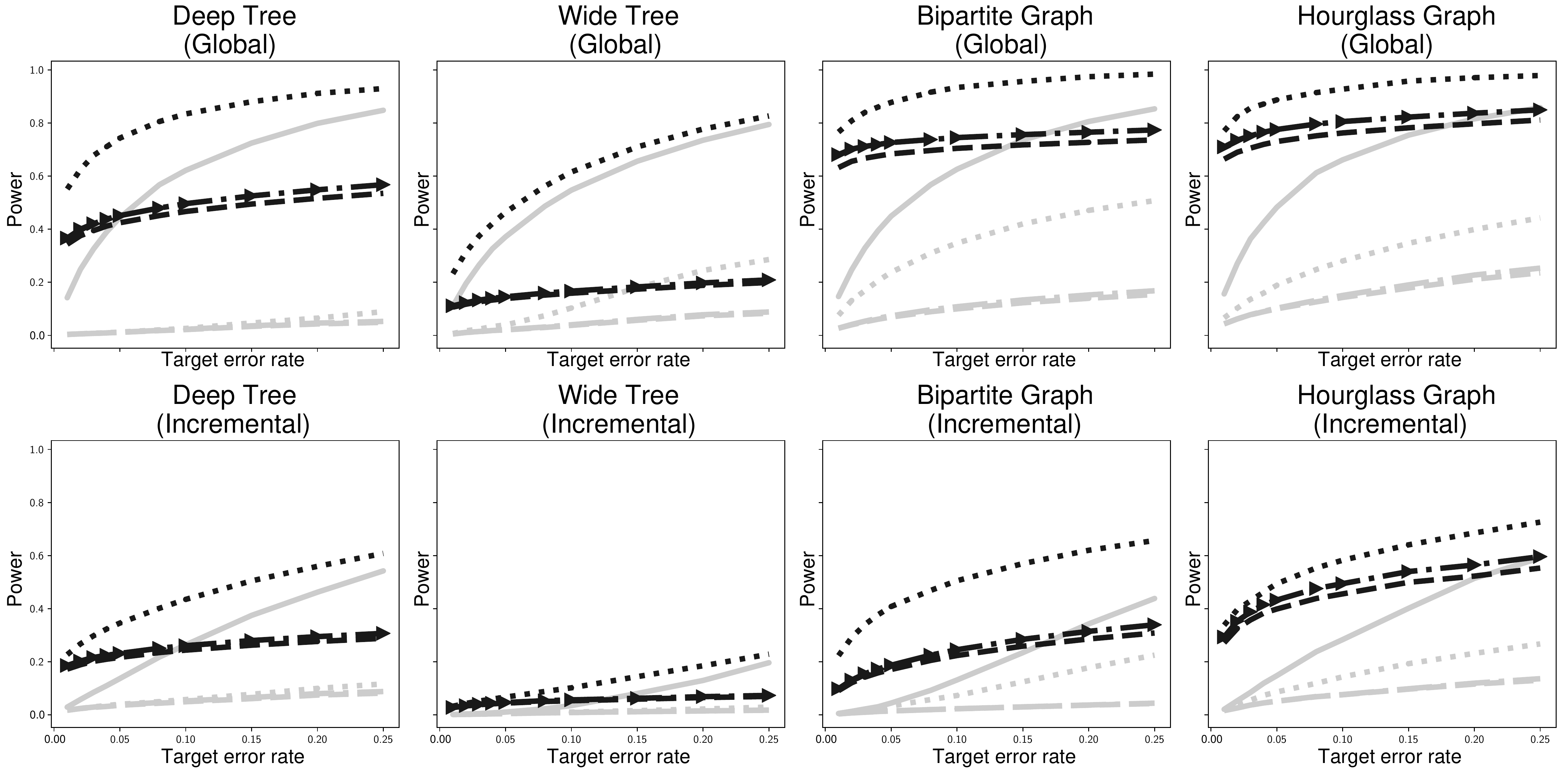}
\caption{\label{fig:simulation-power} Empirical power in each simulation as a function of target error rate. Gray lines are unsmoothed results, black lines are smoothed results; dashed lines use \citet{meijer:goeman:2015:dag-fwer}, dashed lines with arrows use \citet{meijer:goeman:2015:dag-fwer} with Algorithm~\ref{alg:fdx}, and dotted lines use \citet{ramdas:etal:2019:dagger}; the solid gray line is the structureless method of \citet{benjamini:hochberg:1995:bh}.}
\end{figure}

\Cref{fig:simulation-power} presents empirical estimates of power for each simulation. In each scenario, Fisher smoothing boosts the power of all three methods.  Moreover, in all simulations \DAGGER{} actually performs as well or better than \BH{}. This is particularly promising since \citet{ramdas:etal:2019:dagger} found that the \BH{} method almost always had higher power and only in very limited scenarios would the structured method outperform. This is completely reversed in the smoothed case, with the \BH{} method generally having lower power due to being unaware of the structure of the method.  

\begin{figure}[t!]
\includegraphics[width=\linewidth]{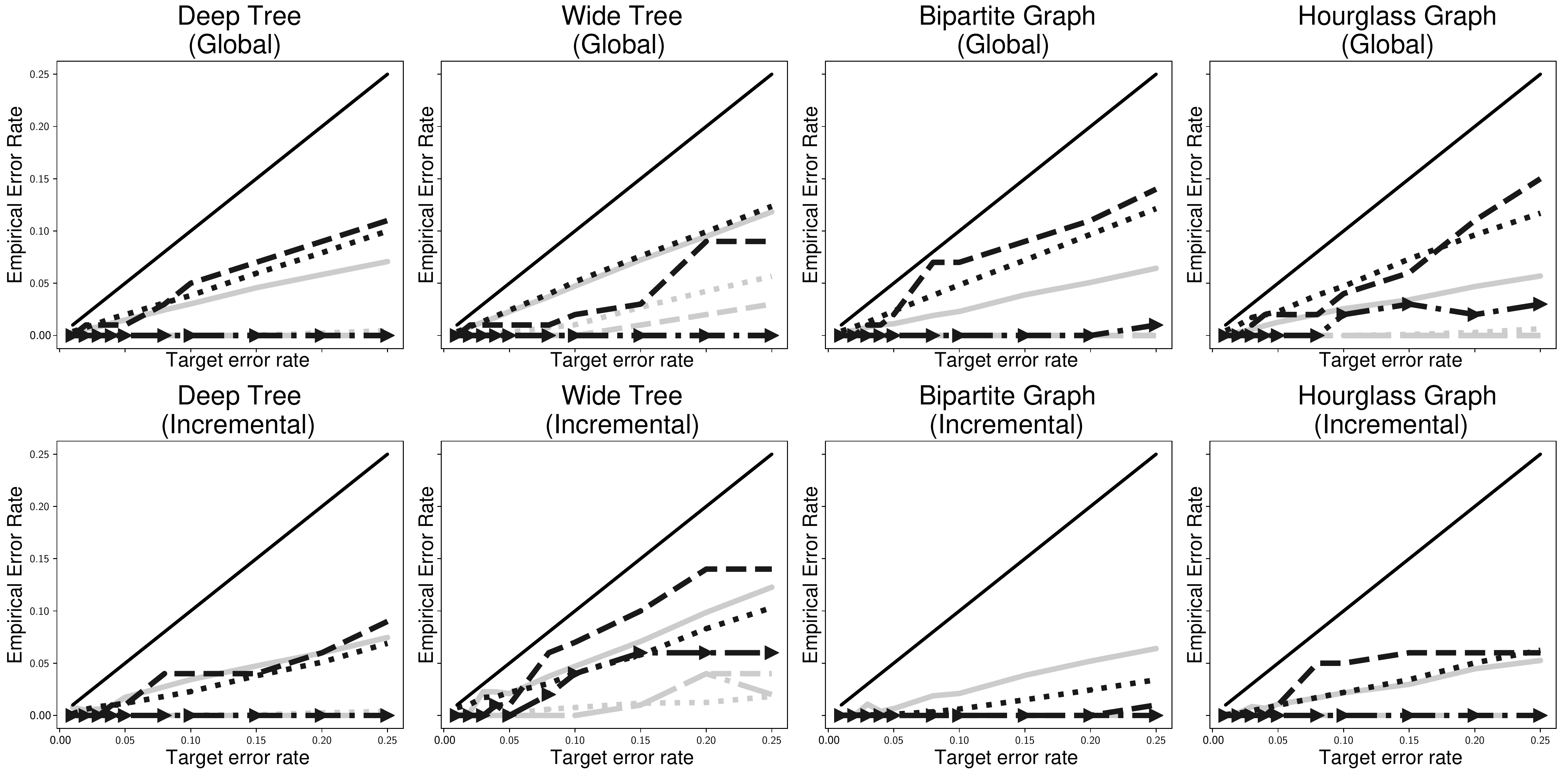}
\caption{\label{fig:simulation-error} Empirical target error rates in each simulation. Lines match those in \cref{fig:simulation-power}; the solid black line is the $(0,1)$ line (maximum allowable error). To facilitate comparison, each method is plotted using its specific target error metric: dashed lines target \fwer{}, dashed lines with arrows target \fdx{}, and solid and dotted lines target \fdr{}.}
\end{figure}

\Cref{fig:simulation-error} confirms that indeed all methods conserve their target error rates empirically. In general, smoothing makes each method less conservative but not to the point of violating the target rate. This is precisely the desired outcome: given a budget for errors, one would prefer to make full use of the budget in order to maximize the number of discoveries.

Fisher smoothing is not guaranteed to increase the power of any of these algorithms, though for any given scenario there is always some form of smoothing which will increase power.  Smoothing methods can help most when the smoothed values for nonnull hypotheses are most heavily influenced by $p$-values from other nonnull hypothesis.  Appendix \ref{subapp:badexperiments} contains intuition and numerical experiments which may help users select smoothing functions which are appropriate for their data.

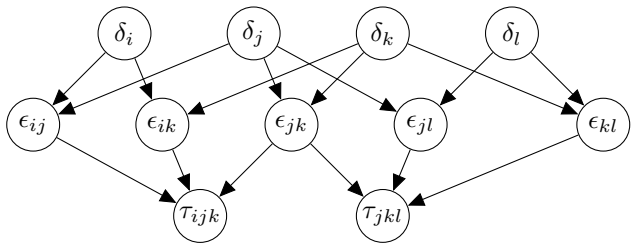
\begin{figure}[t]
\centering
\begin{tikzpicture}

  \node[latent]                               (di) {$\delta_i$};
  \node[latent, right=of di]                               (dj) {$\delta_j$};
  \node[latent, right=of dj]                               (dk) {$\delta_k$};
  \node[latent, right=of dk]                               (dl) {$\delta_l$};
  \node[latent, below left=of di]                (eij) {$\epsilon_{ij}$};
  \node[latent, right=of eij]                (eik) {$\epsilon_{ik}$};
  \node[latent, right=of eik]                (ejk) {$\epsilon_{jk}$};
  \node[latent, right=of ejk]                (ejl) {$\epsilon_{jl}$};
  \node[latent, below right=of dl]                (ekl) {$\epsilon_{kl}$};
  \node[latent, below left=of ejk]                (tijk) {$\tau_{ijk}$};
  \node[latent, below right=of ejk]                (tjkl) {$\tau_{jkl}$};

  \edge {di} {eij} ; %
  \edge {di} {eik} ; %

  \edge {dj} {eij} ; %
  \edge {dj} {ejk} ; %
  \edge {dj} {ejl} ; %
  
  \edge {dk} {eik} ; %
  \edge {dk} {ejk} ; %
  \edge {dk} {ekl} ; %

  \edge {dl} {ejl} ; %
  \edge {dl} {ekl} ; %

  \edge {eij} {tijk} ; %
  \edge {eik} {tijk} ; %
  \edge {ejk} {tijk} ; %

  \edge {ejk} {tjkl} ; %
  \edge {ejl} {tjkl} ; %
  \edge {ekl} {tjkl} ; %

\end{tikzpicture}
\caption{\label{fig:genetic-interaction-dag} Example directed acyclic graph for the genetic interaction study. Individual gene knockouts $\delta$s are the top of the graph, pair knockouts $\tau$ are in the middle, and triplet knockouts $\epsilon$ are the leaves. Each node corresponds to an experiment conducted independently and has an independent $p$-value. If any set of genes potentially contributes to synthetic lethality, the null hypothesis at that node is rejected; all subsets must implicitly be rejected as well.}
\end{figure}

\begin{figure}[t!]
\centering 
\begin{subfigure}{0.48\textwidth}\includegraphics[width=\textwidth]{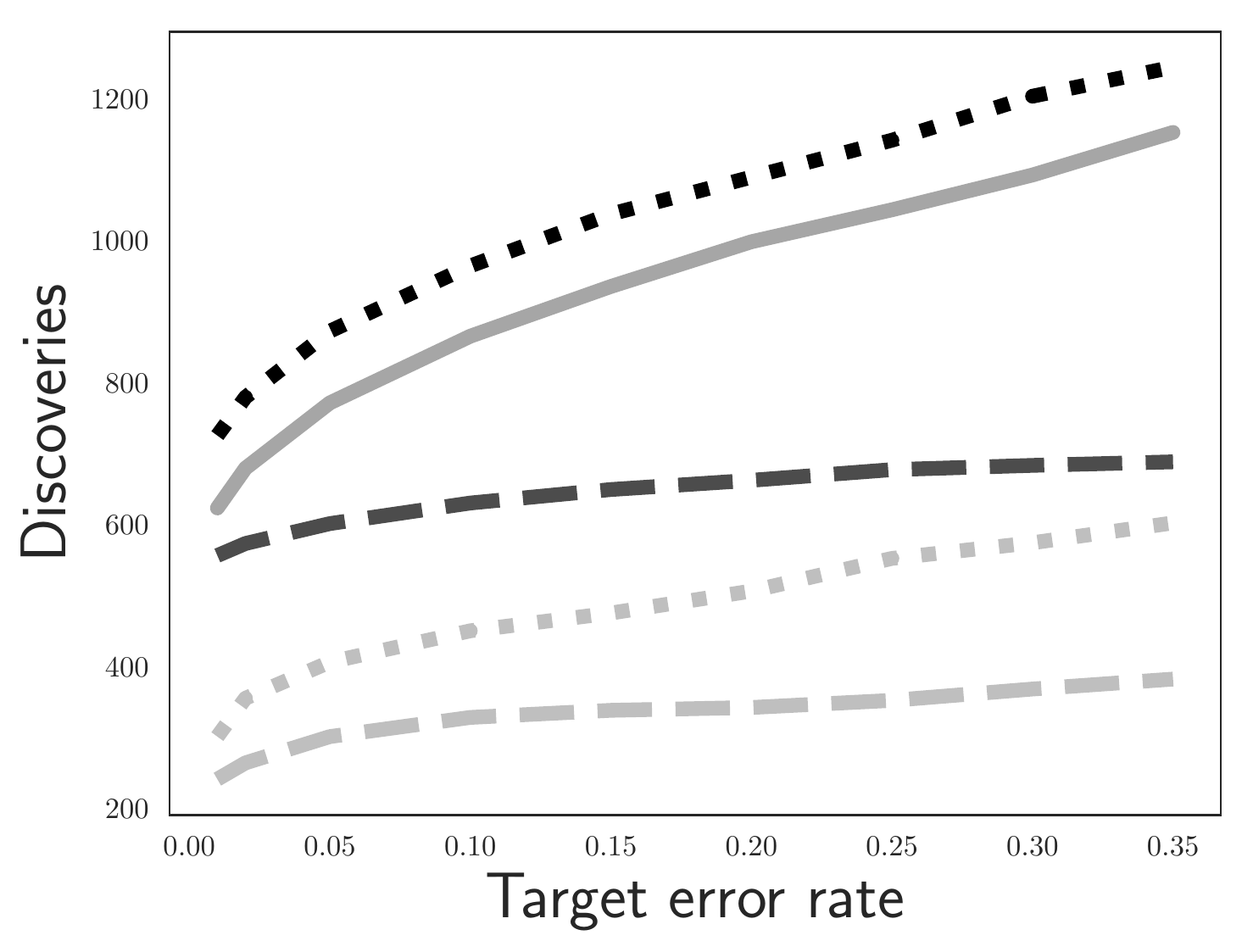}\end{subfigure}
\begin{subfigure}{0.48\textwidth}\includegraphics[width=\textwidth]{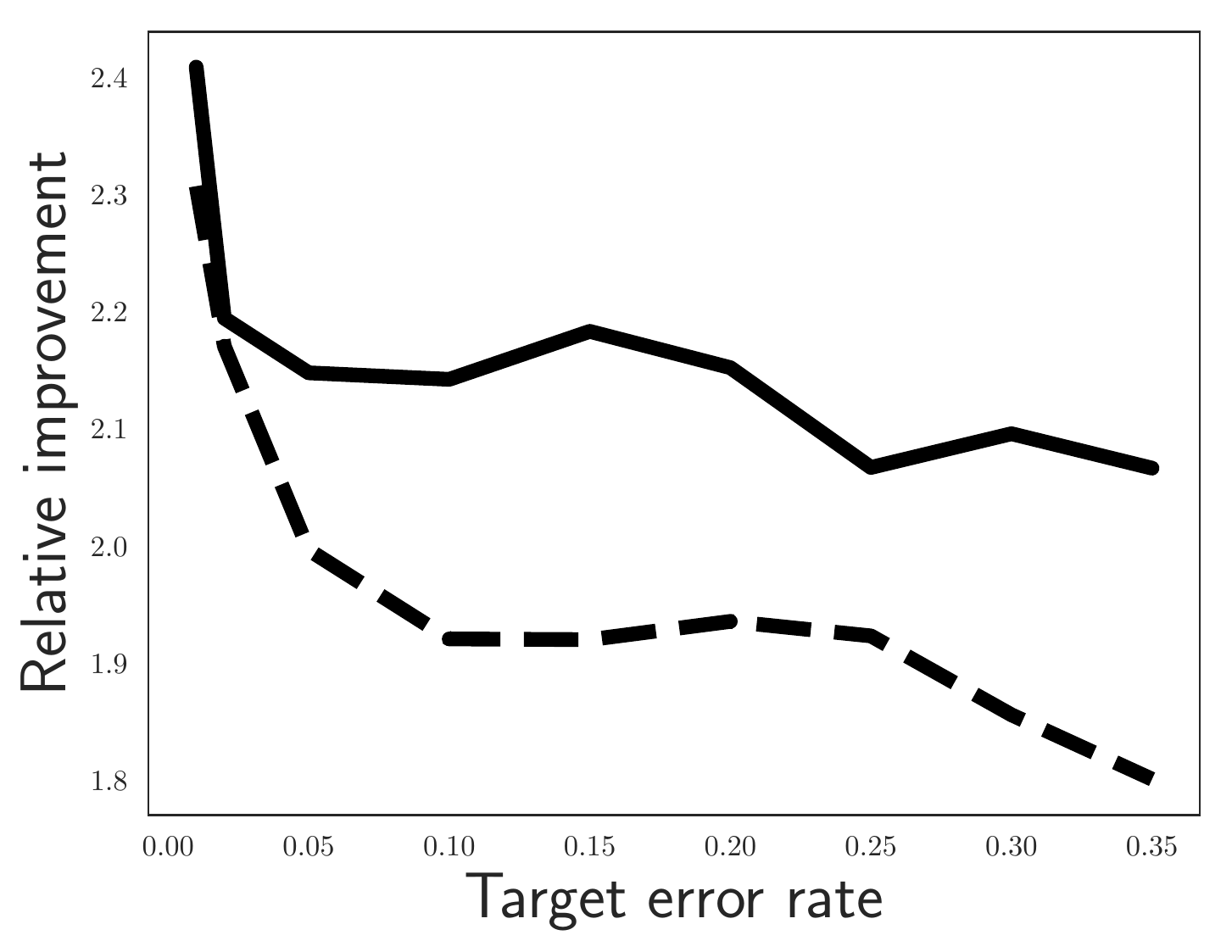}\end{subfigure}
\caption{\label{fig:genes} Performance comparison of raw $p$-values versus smoothed $\psmooth$-values on a biological dataset. Left: total discoveries reported by each method at varying error rates. Solid gray line: \BH{}; dotted black line: Fisher smoothing with \DAGGER{}, dashed dark gray line: Fisher smoothing with the method of \citet{meijer:goeman:2015:dag-fwer}; light gray lines at bottom are the same two methods without smoothing. Right: relative gain of using smoothed $\psmooth$-values over raw $p$-values. Solid black line: relative improvement for \DAGGER{}; dashed black line: relative improvement for the method of \citet{meijer:goeman:2015:dag-fwer}.}
\end{figure}

\subsection{Application to Genetic Interaction Maps}
\label{subsec:results:genes}
We demonstrate the gains of smoothed testing on a real dataset of genetic interactions in yeast cells \citep{kuzmin:etal:2018:yeast-trigenic-interactions}. The data measure the effects of treatments on cell population ``fitness''-- the population size after a fixed incubation window, relative to the initial population size before treatment. The treatment in each experiment is a gene knockout screen that disables a specified set of genes in the population; experiments in the dataset include gene knockout sets of size $1$, $2$, and $3$. For experiments with more than a single gene knocked out, the goal is to determine whether there is any added interaction between the genes that affects fitness. The outcome of interest is the fitness beyond what is expected from independent effects,
\begin{align*}
\label{eqn:genetic_fitness}
\epsilon_{ij} &= \delta_{ij} - (\delta_i \delta_j) \\
\tau_{ijk} &= \delta_{ijk} - (\delta_i \delta_j \delta_k) - \epsilon_{ij} \delta_k - \epsilon_{ik} \delta_j - \epsilon_{jk} \delta_i \, ,
\end{align*}
where $\delta_i$ is the fitness of the population when knocking out the $i^{th}$ gene. The pair score $\epsilon_{ij}$ and triplet score $\tau_{ijk}$ capture the added effect on fitness of knocking out the entire set. In words, $\epsilon_{ij}$ and $\tau_{ijk}$ model the interaction between the genes in the target set above what would be expected by chance if there were no unique interaction between all genes in the set. A set of genes with a negative interaction score is known as a \textit{synthetic lethal} set. See Figure 1 in \citet{kuzmin:etal:2018:yeast-trigenic-interactions} for detailed experimental procedures and details on the definition of $\delta$, $\epsilon$, and $\tau$.

The scientific goal in the yeast dataset is to determine which gene (sub)sets have the potential to contribute to synthetic lethality if disabled. Individual genes may not always reflect this. DNA damage repair mechanisms and other cellular machinery may compensate for an individual knocked out gene to lead to minuscule effects on fitness. If that machinery is also disabled via a second knockout, the signal may become clearer. If a knockout of a pair $(i, j)$ leads to an interaction that produces a synthetic lethal result, the implication is that genes $i$ and $j$ both have the potential to contribute to synthetic lethality.

\cref{fig:genetic-interaction-dag} shows the implied directed acyclic graph encoding the null hypotheses in the yeast dataset. The graph has three levels. Individual gene knockouts form the root nodes, pair knockouts form the middle, and triplet knockouts form the leaves. If any of the leaf nodes is rejected, it implies every constituent pair has the potential to contribute to synthetic lethality.

The yeast dataset has $338$ single-gene experiments, $31092$ pair experiments, and $5451$ triplet experiments. This leads to a graph with $36881$ nodes and $78519$ edges. Each experiment is conducted independently across four replicates. A $t$-test is run for each target outcome variable ($\delta_i$, $\epsilon_{ij}$, or $\tau_{ijk}$) to compare the mean population size to the expected population size.

We use Fisher smoothing and perform selection with \fwer{} control via \citet{meijer:goeman:2015:dag-fwer} and \fdr{} control via \citet{ramdas:etal:2019:dagger}, each at the target error levels,
\[\alpha=(0.01, 0.02, 0.05, 0.1, 0.15, 0.2, 0.25, 0.3, 0.35).
\]
\Cref{fig:genes} shows the results and comparison to the same methods on the original test statistics. As in the simulations, the power gains are substantial: between $1.8$x and $2.4$x more discoveries after smoothing. Further, the smoothed \DAGGER{} power is slightly higher than \BH{}; without smoothing, \DAGGER{} would be substantially lower power than \BH{}. Smoothing therefore has the important effect of recovering or even surpassing the power of \BH{} while preserving the logical nesting structure.

%% file: simulations.tex
To benchmark the power gains for smoothing, we run a set of simulations under different graph structures, alternative hypotheses, and target error metrics. In each case, we use Fisher smoothing on descendants. We consider the following \DAG{} structures:
\begin{itemize}
\item Deep tree. A tree graph with depth $8$ and branching factor $2$.
\item Wide tree. A tree graph with depth $3$ and branching factor $20$.
\item Bipartite graph. A two-layer graph with $100$ roots and $100$ leaves. Each root is randomly connected to $20$ leaves.
\item Hourglass graph. A three-layer graph with $30$ roots, $10$ middle nodes, and $30$ leaves. Each (root, middle) and (middle, leaf) edge is added with probability $0.2$. The graph is then post-processed to ensure each node is connected to at least one node, with middle nodes having at least one incoming and one outgoing node.
\end{itemize}
We generate one-sided $p$-values from $z$-scores with the null $z$-scores drawn from a standard normal. For each structure, we consider two scenarios:
\begin{itemize}
    \item Global alternative. The alternative distribution at each nonnull node is $\mathcal{N}(2, 1)$. \DAG{}s are populated starting at the leaves and null nodes are flipped to nonnull with probability $0.5$.
    \item Incremental alternative. The alternative distribution at each nonnull node is $\mathcal{N}(1+0.3\times (D-d), 1)$, where $d$ is the depth of the node and $D$ is the maximum depth of the \DAG{}. The graph is populated starting at the leaves with nonnull probability $0.5$ and internal nodes are intersection hypotheses that are null if and only if all their child nodes are null.
\end{itemize}

%% file: discussion.tex
\section{Discussion}
\label{sec:discussion}
\subsection{Reshaping for false discovery rate control under dependence}
\label{subsec:discussion:reshaping}
When targeting control of the \fdr{}, we relied on \DAGGER{} for selection after smoothing. There are actually two variants of this method, each extending structureless testing methods to the \DAG{} testing scenario. We focused on the version which extends \BH{} and requires \PRDS{}. Another version extends the \citet{benjamini:yekutieli:2001:fdr-dependence} procedure to \DAG{}s by reshaping the node-level test statistics.

As in the structureless procedure, reshaping controls the \fdr{} regardless of any dependency structure among the statistics, so long as they are marginally (super-)uniform. This reshaping procedure can be applied directly to the smoothed $\psmooth$-values to control the \fdr{}. However, the reshaping procedure raises the bar for rejection, often leading to low power, and will always lead to strictly lower power than the \BH{} extension. Nevertheless, the reshaping variant could be used to control the \fdr{} when using smoothing functions for which no \PRDS{} guarantee is available. In these cases, \fdr{} control would be feasible if the smoothing functions led to $\psmooth$-values that were marginally super-uniform; proving this for a given smoothing function may require some knowledge of the dependency structure.

\subsection{More powerful procedures for false exceedance control}
In section \ref{subsec:testing:fdx}, we proposed a greedy algorithm, outlined in Appendix \ref{subapp:fdx}, by combining the method of \citet{meijer:goeman:2015:dag-fwer} and the generic procedure of \cite{genovese:wasserman:2006:fdx}. Since any data-dependent topological ordering suffices, it is ideal to choose one that yields the highest power. Intuitively, we could iteratively add the most ``promising'' hypothesis, like the one with smallest p-values, which does not break the logical constraint. More generally, we could move beyond greedy algorithms by defining a loss function on each subset in $\mathcal{V}\setminus \rejected_0$ with cardinality $|\rejected_0| \gamma / (1 - \gamma)$ and finding one that minimizes the loss via a combinatorial optimization algorithm. For instance, the loss function can be defined as the sum of p-values. 


\subsection{Ancestor smoothing functions}
\label{subsec:discussion:ancestors}
We described a general framework for smoothing test statistics on \DAG{}s by sharing statistical strength between related nodes. Our current work focused on descendant smoothing functions, which ignore the information contained in ancestor nodes. An alternative class of smoothing functions uses ancestor nodes to smooth the node statistics. These ancestor smoothing functions use prior knowledge about the experiment to adapt their test statistic based on the data in ancestral nodes. Ancestor smoothing functions construct $\psmooth_v$ by fixing all ancestor $p$-values of $p_v$ and fitting a model to the ancestor values. The model requires prior knowledge of the alternative hypothesis, such as knowing that the alternative signal attenuates with the depth of the graph. We expect this to be the case for many real-world scenarios where shallower nodes represent more complex mechanisms or stronger interventions.

Ancestor functions have the appeal of potentially incorporating prior knowledge to gain higher power, but come with substantial trade-offs. First, they typically do not have a closed form null distribution. This makes them computationally expensive, as they require many Monte Carlo simulations in which for every trial the adaptive procedure must be re-run. They are also much more difficult to analyze theoretically, since the the joint distribution of ancestor-smoothed $p$-values involves contributions from variables associated with false hypotheses.   Pragmatically, we have not found any practical examples where the prior knowledge is so strong that it leads to meaningful increases in performance over descendant smoothing. We leave investigation of ancestor smoothing and hybrid ancestor-descendant smoothing functions to future work.

%% file: appendix.tex

\section{A general theory of \PRDS{}}\label{sec:PRDS}

We establish a general theory of \PRDS{} that extends the results of \cite{benjamini:yekutieli:2001:fdr-dependence} to a much broader class of statistics. Section \ref{subapp:generic} presents a generic result that connects the concepts of stochastic ordering with \PRDS{}. In Section \ref{subapp:average} and \ref{subapp:order}, we prove results which imply Theorem \ref{thm:combination_test} and Theorem \ref{thm:order_test}, respectively. Finally, the technical lemmas are presented in Section \ref{subapp:lemmas}. 

\subsection{A Generic Result}\label{subapp:generic}
We start by two definitions related to stochastic orderings. 

\begin{definition}\label{def:ordering} 
  Given two vectors $x, y\in \R^{n}$, $x\preceq y$ iff $x_{i}\le y_{i}$ for all $i\in \{1, \ldots, n\}$. A function $f: \R^{d}\mapsto \R$ is non-decreasing iff $f(x)\le f(y)$ for any $x\preceq y$. Further, given two probability distributions $P_{1}, P_{2}$ on $\R^{d}$, $P_{1}\preceq P_{2}$ iff for any nondecreasing function $f$,
\[\int_{\R^{d}}fdP_{1}\le \int_{\R^{d}}fdP_{2}.\]
\end{definition}

\begin{definition}\label{def:ust}
  A random vector $X$ is said to be stochastically increasing  in the random variable $Y$, denoted by $X\ust Y$, if the regular conditional probability $\P(X \in \cdot \mid Y = y)$ exists (e.g. the underlying measurable space is a Polish space), and $\E [g(X) \mid Y = y]$ is bounded nondecreasing in $y$ for every nondecreasing function $g$. 
\end{definition}

Indeed, $(X_{1}, \ldots, X_{n})$ satisfy \PRDS{} on a subset $T\subset \{1, \ldots, n\}$ iff
\[(X_{1}, \ldots, X_{n})\ust X_{i}, \quad \forall i\in T.\]
The ``if'' part is straightforward because the function $x\mapsto I(x\in \mathcal{C})$ for any non-decreasing set is non-decreasing. The ``only if'' part can be proved by approximating each non-decreasing function by sums of indicator functions on non-decreasing set; See Theorem 1 of \cite{kamae1977stochastic} (Proposition \ref{prop:coupling}).

The following theorem provides a broad class of multivariate statistics that satisfy \PRDS{}. The intuition is quite simple: under the condition $(X_{1}, \ldots, X_{n})\ust Y_{j}$, that $Y_{j}$ ``increases'' would imply that $(X_{1}, \ldots, X_{n})$ ``increases'' and thus all other $Y_{i}$'s because $f_{i}$'s are non-decreasing. 

\begin{theorem}\label{thm:PRDS_generic}
  Let $X_{1}, \ldots, X_{n}$ be random variables on $\R$ and $\{f_{j}: j \in [M]\}$ be a collection of (entrywise) non-decreasing functions on $\R^{n}$. For each $j \in [M]$, define 
\[Y_{j} = f_{j}(X_{1}, \ldots, X_{n}).\]
Then $(Y_{1}, \ldots, Y_{M})$ is \PRDS{} on the subset $\{j\}$ if 
\[(X_{1}, \ldots, X_{n})\ust Y_{j}.\]
\end{theorem}
\begin{proof}
  It is left to prove that, for any $y_{1} < y_{2}$ and a bounded non-decreasing function $h$ on $\R^{M}$, 
  \begin{equation}
    \label{eq:goal}
    \E [h(Y_{1}, \ldots, Y_{M})\mid Y_{j} = y_{1}]\le \E [h(Y_{1}, \ldots, Y_{M})\mid Y_{j} = y_{2}].
  \end{equation}
Because $\R^{n}$ is a Polish space, the regular conditional probability\\ $\P((X_{1}, \ldots, X_{n})\in \cdot \mid f_{j}(X_{1}, \ldots, X_{n}) = y)$ exists. Let $P_{1}$ and $P_{2}$ denotes two measures with
\[P_{k}(\cdot) = \P((X_{1}, \ldots, X_{n})\in \cdot \mid f_{j}(X_{1}, \ldots, X_{n}) = y_{k}) ,\quad k = 1, 2.\]
Our condition ensures that $P_{1}\preceq P_{2}$. By Proposition \ref{prop:coupling}, there exists $X_{1}^{(1)}, \ldots, X_{n}^{(1)}$ and $X_{1}^{(2)}, \ldots, X_{n}^{(2)}$ such that 
\[(X_{1}^{(1)}, \ldots, X_{n}^{(1)})\preceq (X_{1}^{(2)}, \ldots, X_{n}^{(2)}),\quad a.s.,\]
and for any Borel set $A\subset \R^{n}$, 
\[\P\lb (X_{1}^{(k)}, \ldots, X_{n}^{(k)})\in A\rb = P_{k}(A) = \P((X_{1}, \ldots, X_{n})\in A \mid Y_{j} = y_{k}).\]
Let 
\[Y_{i}^{(k)} = f_{i}(X_{1}^{(k)}, \ldots, X_{n}^{(k)}), \quad \forall i\in [M], k = 1, 2.\]
Since each $f_{i}$ is non-decreasing, we have
\begin{equation}
  \label{eq:Yj1Yj2}
  Y_{i}^{(1)}\le Y_{i}^{(2)}, \quad \forall i\in [M], \,\, a.s.,
\end{equation}
and for any Borel set $B\subset \R^{M}$ and $k = 1, 2$,
\begin{equation}\label{eq:Ydist}
  \P\lb (Y_{1}^{(k)}, \ldots, Y_{M}^{(k)})\in B\rb = \P((Y_{1}, \ldots, Y_{M})\in B \mid Y_{j} = y_{k}).
\end{equation}
Finally, since $h$ is non-decreasing, \eqref{eq:Yj1Yj2} implies that
\[h(Y_{1}^{(1)}, \ldots, Y_{M}^{(1)})\le h(Y_{1}^{(2)}, \ldots, Y_{M}^{(2)}), \quad a.s.,\]
which further implies that
\[\E h(Y_{1}^{(1)}, \ldots, Y_{M}^{(1)})\le \E h(Y_{1}^{(2)}, \ldots, Y_{M}^{(2)}).\]
The proof of \eqref{eq:goal} is then completed by \eqref{eq:Ydist}.
\end{proof}

\subsection{Weighted averages of independent log-concave random variables}\label{subapp:average}

We first present a powerful result proved by \cite{efron1965increasing}. 
\begin{proposition}\label{prop:efron}[Theorem 1 of \cite{efron1965increasing}; see also \cite{block1985concept}]
  Suppose $X_{1}, \ldots, X_{n}$ are independent random variables on $\R$ or $\Z$ with (potentially distinct) log-concave densities. Then 
  \[(X_{1}, \ldots, X_{n})\ust \sum_{i=1}^{n}X_{i}.\]
\end{proposition}

\begin{remark}
  \cite{efron1965increasing} assumes each $X_{i}$ has a PF2 density, which is equivalent to a log-concave density; see e.g. \cite{balabdaoui2014chernoff}. 
\end{remark}

\begin{theorem}\label{thm:PRDS_average}
 Let $X_{1}, \ldots, X_{n}$ be independent but not necessarily identically distributed random variables on $\R$ or $\Z$. Assume that $T\subset \{1, \ldots, n\}$ such that $X_{i}$ has a log-concave density for each $i\in T$. Let
 \begin{equation}
   \label{eq:weighted_average}
   Y_{1} = g\lb\sum_{i\in T}a_{i}X_{i}\rb,
 \end{equation}
where $a_{i}\ge 0$ for all $i\in T$ and $g$ is non-decreasing. For any $j = 2, \ldots, M$, let 
\[Y_{j} = f_{j}(X_{1}, \ldots, X_{n})\]
for some (entrywise) non-decreasing function $f_{j}$. Then $(Y_{1}, \ldots, Y_{M})$ are \PRDS{} with respect to $\{1\}$.
\end{theorem}
\begin{proof}
  By Theorem \ref{thm:PRDS_generic}, it is left to show that
\[(X_{1}, \ldots, X_{n})\ust Y_{1}.\]
Since $g$ is non-decreasing, this is equivalent to 
\begin{equation}
  \label{eq:goal_sum}
  (X_{1}, \ldots, X_{n})\ust \sum_{i=1}^{n}a_{i}X_{i}
\end{equation}
Let 
\[\td{X}_{i} = b_{i}X_{i},\quad \mbox{where }b_{i} = \left\{
    \begin{array}{cc}
      a_{i} & (a_{i}\not = 0)\\
      1 & (a_{i} = 0)
    \end{array}
\right.,\] 
and define 
\[\td{f}_{j}(x_{1}, \ldots, x_{n}) = f_{j}\lb\frac{x_{1}}{b_{1}}, \ldots, \frac{x_{n}}{b_{n}}\rb\]
Since all $a_{i}$'s are non-negative and $b_{i}$'s are positive, each $\td{f}_{j}$ is non-decreasing. Also $\td{X}_{i}$ has a log-concave density since $X_{i}\mapsto \td{X}_{i}$ is a linear mapping. Let $A = \{i: a_{i} \not = 0\}$, then 
\[\sum_{i=1}^{n}a_{i}X_{i} = \sum_{i\in A}\td{X}_{i}.\]
By Proposition \ref{prop:efron}, 
\[(\td{X}_{i})_{i\in A}\ust Y_{1}.\]
Since $\td{X}_{i}$'s are independent, by Lemma \ref{lem:ust_concat}, 
\[(\td{X}_{1}, \ldots, \td{X}_{n})\ust Y_{1}\Longrightarrow (X_{1}, \ldots, X_{n})\ust Y_{1}.\]
The equation \eqref{eq:goal_sum} is then proved.
\end{proof}

Now we prove Theorem \ref{thm:combination_test} as a special case of Theorem \ref{thm:PRDS_average}.
\begin{proof}[\textbf{(Theorem \ref{thm:combination_test})}]
  Fix any null node $v$. For each node $c\in \mathcal{C}_v$, let $X_c = H_{v, c}^{-1}(p_c)$. Since $H_{v, c}$ is monotone increasing and $p_c\sim \mathrm{Uniform}([0, 1])$, for any $x\in \R$, 
  \[\P(X_c \le x) = \P(p_c\le H_{v, c}(x)) = H_{v, c}(x).\]
  The density function of $p_c$ is then $H_{v, c}'$, which is log-concave. Since $G_v$ and $H_{v, c}^{-1}$ are monotone increasing for any $(v, c)$, each smoothed p-value and original p-value is a monotone increasing transformation of $(X_1, \ldots, X_n)$. Setting $T = \mathcal{C}_b, g = G_v$ and $a_c = 1$ for all $c\in T$ in Theorem \ref{thm:PRDS_average}, we conclude that the smoothed p-values satisfy \PRDS{} on the subset $\{v\}$. Since this holds for any null node $v$, th p-values satisfy \PRDS{} on the subset of nulls.
\end{proof}

\subsection{Local order statistics}\label{subapp:order}

We first present a powerful result proved by \cite{block1987probability}. 
\begin{proposition}[Corollary in Section 3 of \cite{block1987probability}]\label{prop:block1987}
  Let $X_{1}, \ldots, X_{n}$ be independent  random variables on $\R$ with a common continuous distribution. Let $X_{(1)}\le X_{(2)}\le \cdots \le X_{(n)}$ be the ordered statistics. For any given subsets $\{k_{1}, \ldots, k_{r}\}$ of $\{1, \ldots, n\}$, 
\[(X_{1}, \ldots, X_{n})\ust (X_{(k_{1})}, \ldots, X_{(k_{r})}).\]
\end{proposition}

Theorem \ref{thm:order_test} is a direct consequence of the following theorem.
\begin{theorem}\label{thm:PRDS_order}
  Let $X_{1}, \ldots, X_{n}$ be independent and identically distributed random variables with a common continuous distribution. Let 
  \begin{equation}
    \label{eq:ordered_stat}
    Y_{1} = g\lb X_{(k); T}\rb,
  \end{equation}
where $g$ is an non-decreasing function, $T$ is a subset of $[n]$ and $X_{(k); T}$ denotes the $k$-th order statistics of $(X_{i})_{i\in T}$. For any $j = 2, \ldots, M$, let 
\[Y_{j} = f_{j}(X_{1}, \ldots, X_{n})\]
for some non-decreasing function $f_{j}$. Then $(Y_{1}, \ldots, Y_{M})$ are PRDS with respect to $\{1\}$.
\end{theorem}
\begin{proof}
  Similar to \eqref{eq:goal_sum} in proof of Theorem \ref{thm:PRDS_average}, it is left to prove that 
  \begin{equation*}
    (X_{1}, \ldots, X_{n})\ust X_{(k); S}.
  \end{equation*}
 Since $X_{1}, \ldots, X_{n}$ are independent, by Lemma \ref{lem:ust_concat}, we can assume that $T = [n]$ without loss of generality, i.e.
 \begin{equation}
   \label{eq:goal_order}
   (X_{1}, \ldots, X_{n})\ust X_{(k)}.
 \end{equation}
This is guaranteed by Proposition \ref{prop:block1987}.
\end{proof}

\subsection{Technical lemmas}\label{subapp:lemmas}
We first present a coupling property of stochastic orderings.

\begin{proposition}[Theorem 1 of \cite{kamae1977stochastic}]\label{prop:coupling}
  Let $E$ be a partially ordered Polish space (i.e. complete separable metrizable topological space). For any two distributions $P_{1}$ and $P_{2}$ on $E$, $P_{1}\preceq P_{2}$ iff there exists a distribution $P$ on $E \times E$ equipped with the product topology, which is supported on the set $\{(x, y)\in E\times E\}$ with first marginal $P_{1}$ and second marginal $P_{2}$. Equivalently, there exists a random vector $(X, Y)$ on $E\times E$, such that $X \preceq Y$ almost surely.
\end{proposition}

Next we prove a useful property of stochastic monotonicity.

\begin{lemma}\label{lem:ust_concat}
Assume that the sample space of the random element $(X, Y, Z)$ is partially ordered. If $X\ust Y$ and $Z\perp (X, Y)$, then
\[(X, Z)\ust Y, \quad (X, Z)\ust (Y, Z)\]  
\end{lemma}
\begin{proof}
  Let $h$ be any bounded non-decreasing function on the domain of $(X, Z)$. Since $Z \perp (X, Y)$, 
\[\E [h(X, Z)\mid Y] = \E [\E_{Z}[h(X, Z)]\mid Y] = \E \left[\int h(X, z)dP_{Z}(z)\mid Y\right].\]
Let $\td{h}(x) = \int h(x, z)dP_{Z}(z)$. Since $h(x, z)$ is non-decreasing in $x$ for any $z$, $\td{h}(x)$ is also bounded non-decreasing in $x$. Since $X\ust Y$, $\E [\td{h}(X)\mid Y = y]$ is non-decreasing in $y$. This proves the first result. 

For the second result, given any $y, z$,
\begin{align*}
  &\E [h(X, Z)\mid Y = y, Z = z] =  \E [h(X, z)\mid Y = y, Z = z] = \E [h(X, z)\mid Y = y].
\end{align*}
For any $y\preceq y'$ and $z \preceq z'$, 
\[\E [h(X, z)\mid Y = y]\le \E [h(X, z')\mid Y = y]\le \E [h(X, z')\mid Y = y']\]
where the first inequality uses the fact that $h(X, z)$ is non-decreasing in $z$, and the second inequality uses the fact that $h(x, z')$ is non-decreasing in $x$ and $X\ust Y$. Therefore, 
\[\E [h(X, Z)\mid Y = y, Z = z]\le \E [h(X, Z)\mid Y = y', Z = z'].\]
\end{proof}

\section{Other Technical Details}\label{sec:proofs}

We here recall some notation for the benefit of the reader. For each $v\in\{1,\cdots, n\}=\mathcal{V}$ we can observe a random variable $p_v$.  Let $\bar{\mathcal{S}}\subset\mathcal{V}$; we assume $\{p_v\}_{v\in\bar{\mathcal{S}}}\overset{\mathrm{i.i.d}}{\sim}\mathrm{Uniform}[0,1]$.  $\mathcal{G} = (\mathcal{V}, \mathcal{E})$ whose edges encode logical constraints on the hypotheses: $v \rightarrow w \implies (v\in\bar{\mathcal{S}}\implies w\in\bar{\mathcal{S}})$.  For each $v$ we have that $\mathcal{C}_v$ denotes the union of $v$ with all of its descendants in the graph $\mathcal G$ and $f_v(x_{\mathcal{C}_v},x_{\mathcal{V}\backslash\mathcal{C}_v})$ denotes any function. The smoothed $\psmooth$-value for node $v$ is then,
\[
    F_v(c;x_{\mathcal{V}\backslash\mathcal{C}_v}) \triangleq \mathrm{pr}(f_v( u_{\mathcal{C}_v},x_{\mathcal{V}\backslash\mathcal{C}_v}) \leq c),\quad \psmooth_v \triangleq F_v(f_v(p_{\mathcal{C}_v},p_{\mathcal{V}\backslash\mathcal{C}_v});p_{\mathcal{V}\backslash\mathcal{C}_v}),
\]
where $\{ u_w\}_{w\in\mathcal{C}_v}$ \iid{} $\mathrm{Uniform}[0,1]$.

\subsection{Smoothed test statistics}

The starting point for this work is \cref{lem:smoothp}, which we prove here.

\begin{proof}
  Fix any $v\in\bar{\mathcal{S}}$. The meaning of the graph structure indicates that $\mathcal{C}_v \subset \bar{\mathcal{S}}$
  
  \begin{enumerate}[(a)]
  \item Our assumption that the null p-values are \iid{} as $\mathrm{Uniform}[0, 1]$ and independent of non-null p-values imply that
    \[p_{\mathcal{C}_v} \stackrel{d}{=} u_{\mathcal{C}_v}\mid p_{\mathcal{V}\backslash\mathcal{C}_v}.\]
    As a result,
    \[\td{p}_v\stackrel{d}{=}F_v(f_v(u_{\mathcal{C}_v}; p_{\mathcal{V}\backslash\mathcal{C}_v}))\mid p_{\mathcal{V}\backslash\mathcal{C}_v}.\]
    By definition, $F_v(f_v(u_{\mathcal{C}_v}; p_{\mathcal{V}\backslash\mathcal{C}_v}))$ is superuniform.
  \item Since null p-values are independent and superuniform , 
    \[p_{\mathcal{C}_v}\preceq u_{\mathcal{C}_v}\mid p_{\mathcal{V}\backslash\mathcal{C}_v},\]
    where $\preceq$ denotes entrywise stochastic dominance. Since $f_v$ is nondecreasing in $p_{\mathcal{C}_v}$, 
    \[f_v(p_{\mathcal{C}_v}, p_{\mathcal{V}\backslash\mathcal{C}_v})\succeq f_v(u_{\mathcal{C}_v}, p_{\mathcal{V}\backslash\mathcal{C}_v})\mid p_{\mathcal{V}\backslash\mathcal{C}_v}.\]
    Further, since $F_v$ is nondecreasing,
    \[\td{p}_v\succeq F_v(f_v(u_{\mathcal{C}_v}, p_{\mathcal{V}\backslash\mathcal{C}_v}))\mid p_{\mathcal{V}\backslash\mathcal{C}_v}.\]
    As in case (a), we conclude that $\td{p}_v$ is superuniform.
  \end{enumerate}
\end{proof}

\subsection{Familywise error rate control}\label{subapp:fwer}

We apply the method of \cite{meijer:goeman:2015:dag-fwer} for familywise error rate control based on the Sequential Rejection Principle developed in \citet{goeman:solari:2010:sequentialrp}.  For the benefit of the reader, we sketch the main ideas below.  

The starting-place for this analysis is an application of the Sequential Rejection Principle of Goeman and Solari.  

\begin{definition}
Fix $n \in \mathbb{N}$, let $\mathcal{P}$ denote the power set $\{1\cdots n\}$, let $\pi:\ \{1\cdots n\} \times \mathcal{P} \rightarrow \mathbb{R}$.  Let $\rejected_0 \triangleq \emptyset$ and recursively define
\begin{align}
\rejected_i(p,\pi) \triangleq \rejected_{i-1}(p,\pi) \cup \{v: p_v \le  \pi(v,\rejected_{i-1}(p,\pi))\}
\end{align}
The limit of this process, $\rejected(p,\pi) \triangleq \cup_i^n \rejected_i(p,\pi)$, is said to be the final result of sequential rejection on $p$ via $\pi$.  
\end{definition}
Intuitively, $\pi$ is an object which uses the current set of rejections to produce a weight for each hypothesis.  These weights are then used to decide whether further rejections may be made.  Goeman and Solari provide conditions on $\pi$ such that this algorithm conserves familywise error.

\begin{theorem}[Goeman and Solari] \label{thm:goeman-solari}
  Let $\mathbb{M}\subset \mathcal{P}$.  For each $\mathcal{S} \in \mathbb{M}$ let $P_\mathcal{S}$ denote a probability distribution for the variables $p_1,p_2 \cdots p_n \in \mathbb{R}$.   Let $\pi:\ \{1 \cdots n\} \times \mathcal{P}\rightarrow \mathbb{R}$ satisfy $P_\mathcal{S}(\{v:\ p_v \leq \pi(v,\mathcal{S})\} \subset \mathcal{S}) \geq 1-\alpha$ for every $\mathcal{S} \in \mathbb{M}$ (this will be referred to as the ``single-step condition'').  Let us further assume that $\pi(v,\rejected_i(p,\pi)) \leq \pi(v,\mathcal{S})$ for any $i,p$ and any $\mathcal{S}\in\mathbb{M}$ such that $\mathcal{S}\supseteq \rejected_i(p,\pi)$ and any $v\notin \mathcal{S}$ (this will be referred to as the ``monotonicity condition'').  It follows that 
  \[
  P_M(\rejected(p,\pi) \subset \mathcal{S}) \geq 1-\alpha
  \]
  for any $\mathcal{S} \in \mathbb{M}$.

\end{theorem}
  
\begin{proof}
Consider the event that $\{v:\ p_v \leq \pi(v,\mathcal{S})\} \subset \mathcal{S}$, i.e.
\[
p_v > \pi(v,\mathcal{S}) \qquad \forall v\in \mathcal{\bar S}.
\]   
Here we adopt the notation $\mathcal{\bar S} = \{1 \cdots n\} \backslash \mathcal{S}$.  The single-step condition guarantees that this event occurs with probability at least $1-\alpha$.  Thus to prove our point it suffices to show that $\rejected \subset \mathcal{S}$ whenever this event occurs.  

We argue by induction.  Clearly $\rejected_0 \subset M$.  Now suppose that $\rejected_{i-1} \subset \mathcal{S}$; the monotonicity assumption then yields that $\pi(v,\rejected_{i-1}) \leq \pi(v,\mathcal{S})$.  On the other hand, we have already assumed that $p_v >\pi(v,\mathcal{S})$ for every $v \in \mathcal{\bar S}$.  Together, these facts imply that $p_v \geq \pi(v,\rejected_{i-1})$ for every $v \in \mathcal{\bar S}$.  It follows that
\[
\rejected_{i}=\rejected_{i-1} \cup \{v: p_v \le  \pi(v,\rejected_{i-1})\} \subset \mathcal{S}
\]
By induction, it follows that $\rejected \subset \mathcal{S}$. 
\end{proof}

\cite{meijer:goeman:2015:dag-fwer} provide several weight functions that satisfy the single-step and monotonicity conditions of the Theorem above.  For simplicity, in this work we focus on their ``all-parents'' weight function.  This function produces weights by using a particular ``water-filling'' procedure, described below.

\vspace{.1in}
\begin{algorithm}[H]
  \caption{\label{alg:filling} All-parents water-filling procedure}
  \SetKwInOut{Input}{input}\SetKwInOut{Output}{output}
  \Input{a graph $\mathcal{G}=(\mathcal{V},\mathcal{E}$), a subset $\rejected \subseteq \mathcal{V}$}
  $\mathcal{Z}(\rejected) \gets |\{v\notin \rejected:\ v\mbox{ is a leaf}\}|$\;
  \ForEach{$v\in \mathcal{V}$}{
    \eIf{$v$ is a leaf and $v \notin \rejected$}{
      $g_v \gets 1/\mathcal{Z}$\;
    } { 
      $g_v \gets 0$\;
    }
  }
  \While{there exists a node $v$ such that $g_v\geq 0$ and $\mathrm{parents}(v) \backslash \rejected \neq \emptyset$}{
      \ForEach{$w \in \mathrm{parents}(v)\backslash \rejected$} {
        $g_w \gets g_w + \frac{g_v}{|\mathrm{parents}(v)\backslash \rejected|}$\;
      }
      $g_v \gets 0$\;
  }
  \KwRet{$g$}  
\end{algorithm}
\vspace{.1in}

This algorithm can be equivalently defined via the following recursive relations:
\begin{align*}
  \tilde g_v(\mathcal{G},\rejected) &= \begin{cases}
    0 & \mbox{if $v\in \rejected$} \\
    1/|\{v\notin \rejected:\ v\mbox{ is a leaf of $\mathcal{G}$}\}| & \mbox{if $v\notin \rejected$ is a leaf of $\mathcal{G}$}  \\
    \sum_{w \in \mathrm{children}(w)} \frac{\tilde g_w(\mathcal{G},\rejected)}{|\mathrm{parents}(w)\backslash \rejected|}
      & \mathrm{otherwise}
  \end{cases} \\
  g_v(\mathcal{G},\rejected) &= \begin{cases}
    \tilde g_v(\mathcal{G},\rejected) & \mathrm{if}\ \mathrm{parents}(v)\subset \rejected \\
    0 & \mathrm{otherwise}
  \end{cases}
\end{align*}
Both points of view are helpful in the proofs below.

For any directed acyclic graph $\mathcal{G}$ and any given target error level, $\alpha$, this weight-filling function can be used to produce a satisfactory weight function by taking $\pi(v,\rejected) = \alpha g_v(\mathcal{G},\rejected)$.  This weight function is guaranteed to satisfy the conditions of Goeman and Solari.

\begin{theorem}[Meijer and Goeman]
Fix a directed graph $\mathcal{G}$ and let $\mathbb{M} \subset \mathcal{P}$ denote the set of all collections of hypotheses consistent with that graph (i.e.\ if $\mathcal{S} \in \mathbb{M}$ and $v\in \mathcal{S}$ then all ancestors of $v$ are also in $\mathcal{S}$).  Fix any $\mathcal{S} \subset \mathbb{M}$ and assume $P_\mathcal{S}(p_v\leq c)\leq c, \forall v\notin \mathcal{S}$.  Then the weight function $\pi(v,\rejected) = \alpha g_v(\mathcal{G},\rejected)$ satisfies the single-step and monotonicity conditions of Theorem \ref{thm:goeman-solari}.
\end{theorem}
\begin{proof}
  Let us begin with the monotonicity condition. This monotonicity condition follows from the corresponding monotonicity property on the water-filling function $g$, defined above.   Let us fix any $i,p$ and assume  $\rejected_i(p,\pi) \subset \mathcal{S}$.  
  \begin{itemize}
  \item If $v\notin \mathcal{S}$ is a leaf, then $\tilde g_v(\mathcal{G},\rejected_i),\tilde g_v(\mathcal{G},\mathcal{S})$ denote the reciprocal of the number of unrejected leaves in each of their respective cases. Since $S\subset S'$ it follows that $\tilde g_v(\mathcal{G},\rejected_i)\leq \tilde g_v(\mathcal{G},\mathcal{S})$ for all $v \notin \mathcal{S}$.
  \item Now consider the case that $w \notin \mathcal{S}$ isn't a leaf.   Let us apply the inductive hypothesis that $\tilde g_v(\mathcal{G},\rejected_i)\leq \tilde g_v(\mathcal{G},\mathcal{S})$ for every child $v \notin \mathcal{S}$.  Now note that $\mathcal{S}\in \mathbb{M}$ and $w \notin \mathcal{S}$; it follows that all of the children of $w$ also lie outside of $\mathcal{S}$.  Thus, in fact, since all children of $w$ lie outside $\mathcal{S}$, we have that $\tilde g_v(\mathcal{G},\rejected_i)\leq \tilde g_v(\mathcal{G},\mathcal{S})$ for every child $v$.  We obtain that
  \begin{align*}
  \tilde g_w(\mathcal{G},\rejected_i) & =\sum_{v\in\mathrm{children}(w)}\frac{\tilde g_v(\mathcal{G},\rejected_i)}{\left|\mathrm{parents}(v)\backslash \rejected_i\right|}\leq
  \sum_{v\in\text{children}(w)}\frac{\tilde g_v(\mathcal{G},\rejected_i)}{\left|\text{parents}(v)\backslash\mathcal{S}\right|}\leq \tilde g_w(\mathcal{G},\mathcal{S})
  \end{align*}
  \end{itemize}
  Putting these two facts together and applying an induction argument on the nodes in topological order from the bottom of the tree, it follows that $\tilde g_v(\mathcal{G},\rejected_i) \leq \tilde g_v(\mathcal{G},\mathcal{S})$ for every $v\notin \mathcal{S}$.  Finally, observe that
  \[
    \mathrm{parents}(v) \subset \rejected_i \implies \mathrm{parents}(v) \subset \mathcal{S}
  \]
  In combination with the inequality on $\tilde g$, this yields that $g_v(\mathcal{G},\rejected_i) \leq g_v(\mathcal{G},\mathcal{S})$ for all $v\notin \mathcal{S}$, as desired.  

  Now let us turn to the single-step condition.  First observe that the water-filling function $g$ always satisfies $\sum_v g_v =1$.  Indeed, observe that Algorithm \ref{alg:filling} initializes the values of $g_v$ to satisfy this property, and the total value of $\sum_v g_v$ is unchanged at each iteration of the while loop.  The single-step condition then follows from a union bound:
  \begin{align*}
  \mathrm{pr}\left(\{v: p_v \le \pi(v,\mathcal{S})\}  \subseteq \mathcal{S}\right) 
      &= \mathrm{pr}(p_v > \alpha g_v(\mathcal{G},\mathcal{S})\  \forall v \in \bar{\mathcal{S}})  \\
      &\geq 1- \sum_{v \in \bar{\mathcal{S}}} \mathrm{pr}(p_v \leq \alpha g_v(\mathcal{G},\mathcal{S}))\\
      &\geq 1 - \sum_{v \in \bar{\mathcal{S}}} \alpha g_v(\mathcal{G},\mathcal{S}) \geq 1-\alpha.
  \end{align*}
\end{proof}

\subsection{False exceedance control}\label{subapp:fdx}

We combine the method of \citet{meijer:goeman:2015:dag-fwer} and the generic procedure of \cite{genovese:wasserman:2006:fdx}, which turns any familywise error rate controlling method into a false exceedance controlling method. The procedure is outlined in Alg. \ref{alg:fdx}. Here $\texttt{topological\_sort}$ refers to a topological ordering of the vertices of the graph such that any node appears later than its parents \citep[][Section 22.4]{cormen:2001:toposort}.

\input{fdx_algorithm}


\subsection{False discovery rate control}

For false discovery rate control, we use the unreshaped version of \DAGGER{} on our $\psmooth$-values. It is a recursive step-up procedure with graph-specific thresholds. In particular, for a directed acyclic graph, the depth of each node $d_v$ is defined as the length of the longest path from $v$ to any root, the node without parents. Let $\mathcal{V}_d = \{v: d_v = d\}$. For each node $v\in \mathcal{V}_d$, let $\{\alpha_{d, v}(r): r = 1, 2, \ldots\}$ be a sequence of thresholds which will be specified later. The procedure starts by applying a generalized step-up procedure on $\mathcal{V}_1$ to obtain an initial rejection set:
\[\rejected_1 = \left\{v\in \mathcal{V}_1: \td{p}_v\le \alpha_{1, v}(R_1)\right\},\]
where
\[R_1 = \max\left\{1\le r\le |\mathcal{V}_1|: \sum_{v\in \mathcal{V}_1}I(\td{p}_v\le \alpha_{1, v}(r))\ge r\right\}.\]
Once $\rejected_1, \ldots, \rejected_{d-1}$ are decided, let
$\mathcal{V}'_d$ be the set of nodes with depth $d$ and all parents rejected, i.e.
\[\mathcal{V}'_d = \left\{v\in \mathcal{V}_d: \texttt{parents}(v)\subset \bigcap_{j=1}^{d-1}\rejected_j\right\}.\]
Then the procedure computes $\rejected_d$ as
\[\rejected_d = \left\{v\in \mathcal{V}'_d: \td{p}_v\le \alpha_{d, v}(R_d)\right\},\]
where
\[R_d = \max\left\{1\le r\le |\mathcal{V}'_d|: \sum_{v\in \mathcal{V}'_d}I(\td{p}_v\le \alpha_{d, v}(r))\ge r\right\}.\]
The final rejection set of \DAGGER{} is then
\[\rejected = \bigcup_{d=1}^{D}\rejected_d,\]
where $D$ is the maximal depth. By definition, every node in $\rejected_d$ has all their parents rejected in $\rejected_1\cup \ldots \cup \rejected_{d-1}$. Therefore, $\rejected$ obeys the logical constraint.

To define the threshold $\alpha_{d, v}(r)$, \cite{ramdas:etal:2019:dagger} define the effective number of leaves $\ell_v$ and the effective number of nodes $m_v$ for each node $v$. Similar to the weight in Meijer-Goeman algorithm, $\ell_v$ and $m_v$ are computed via a ``water-filling'' algorithm in a bottom-up fashion. Specifically, they set $\ell_v = m_v = 1$  for each leaf node and then proceed up the tree, from leaves to roots, recursively calculating
\[\ell_v = \sum_{w\in \texttt{children}(v)}\frac{\ell_w}{|\texttt{parents}(w)|}, \quad m_v = 1 + \sum_{w\in \texttt{children}(v)}\frac{m_w}{|\texttt{parents}(w)|}.\]
Let $L$ is the total number of leaves. Then
\[\alpha_{d, v}(r) = \alpha \frac{\ell_v}{L}\frac{m_v + r + \sum_{j=1}^{d-1}R_j - 1}{m_v}.\]

\subsection{Dependency}

\label{subapp:depproofs}

Let $\Phi^{-1}$ denote the quantile function of the standard normal.  We here consider the case that $Z_v = \Phi^{-1}(p_v)$ satisfies $Z_{\bar{\mathcal{S}}} \sim \mathcal{N}(0,R)$ for some correlation matrix $R$.  In this case some of the methods presented in the main text may fail due to the correlation between the $p$-values.  However, we can account for this dependency, even if we do not know $R$.  Throughout this section we consider this case, and apply what we call conservative Stouffer smoothing.
\[
\psmooth_v \gets 
    \begin{cases}
        1 & \mathrm{if } \sum_{w \in \mathcal{C}_v} \pi_{vw} Z_w\geq 0 \\
        \Phi\left(\sum_{w \in \mathcal{C}_v} \pi_{vw} Z_w\right) & \mathrm{otherwise},
    \end{cases}
\]
where $\Phi$ denotes the standard normal CDF and $\pi$ satisfies $\pi_{vw}\geq 0,\sum_{w \in \mathcal{C}_v} \pi_{vw}=1$. As in the independent case, the existing techniques of \citet{meijer:goeman:2015:dag-fwer} and \citet{ramdas:etal:2019:dagger} can be now be applied to these smoothed values.  The target error rate will still be controlled, even though the null $p$-values are no longer independent.

\switchforrestateable{mvnfwer}
\begin{proof}
    Fix any $v \in \bar{\mathcal{S}}$.  Let $Y_v = \sum_{w \in \mathcal{C}_v} \pi_{vw} Z_w$.  The assumption of Gaussian copula tells us that $Y_v \sim \mathcal{N}(0,\sigma^2)$.  We can get an upper bound for $\sigma^2$ using the observation that $R_{w,w'} \in [-1,1]$ and $\pi$ lies on the simplex:
    \begin{align*}
    \mathrm{var}(Y_v) &= \sum_w \pi_{w} \left(\sum_{w'} R_{w,w'}\pi_{w'}\right) 
                    \leq \sum_{w} \pi_{w}(1) = 1.
    \end{align*}
    Thus $\sigma \leq 1$.  We can now prove the Lemma, considering three different cases.
    First, let $\alpha \leq 1/2$.  Then
            \begin{align*}
                \mathrm{pr}(\psmooth_v \leq \alpha) 
                    & = \mathrm{pr}(Y_v \leq \Phi^{-1}(\alpha)) = \Phi(\Psi^{-1}(\alpha)/\sigma) \\
                    & \leq \Phi(\Phi^{-1}(\alpha)) = \alpha.
            \end{align*}
            This only works because we can assume $\Phi^{-1}(\alpha)\leq 0$; otherwise the inequality goes the other way.
    Second, let $1/2 < \alpha < 1$.  Then
            \begin{align*}
                \mathrm{pr}(\psmooth_v \leq \alpha) 
                    &= \mathrm{pr}(\psmooth_v \leq 1/2) + \mathrm{pr}(\psmooth_v\in (1/2,\alpha)) \\ 
                    & = \mathrm{pr}(\psmooth_v \leq 1/2) + 0 \leq 1/2 \leq \alpha.
            \end{align*}
     Finally, let $\alpha=1$.  Then $\mathrm{pr}(\psmooth_v \leq 1)=1=\alpha$.
\end{proof}

\switchforrestateable{mvnfdr}

\begin{proof} As in the previous proof, let $Y_v = \sum_{w \in \mathcal{C}_v} \pi_{vw} Z_w$.  Note that $Y_{\bar{\mathcal{S}}}$ is jointly Gaussian.  Since $R$ has no negative entries and $\pi$ is also non-negative, the covariance of $Y_{\bar{\mathcal{S}}}$ is positive.  It follows that $Y$ are \PRDSadj{} \citep{benjamini:yekutieli:2001:fdr-dependence}.  Moreover, $\psmooth$ is an elementwise monotone nondecreasing transformation of $Y$; each $\psmooth_v$ can be expressed as a monotone nondecreasing function of $Y_v$; it follows that $\psmooth$ are also \PRDSadj{}.\end{proof}

\section{Additional experimental results}

\subsection{Dependent null statistics}

\label{subapp:depexperiments}

\begin{figure}
  \includegraphics[width=\linewidth]{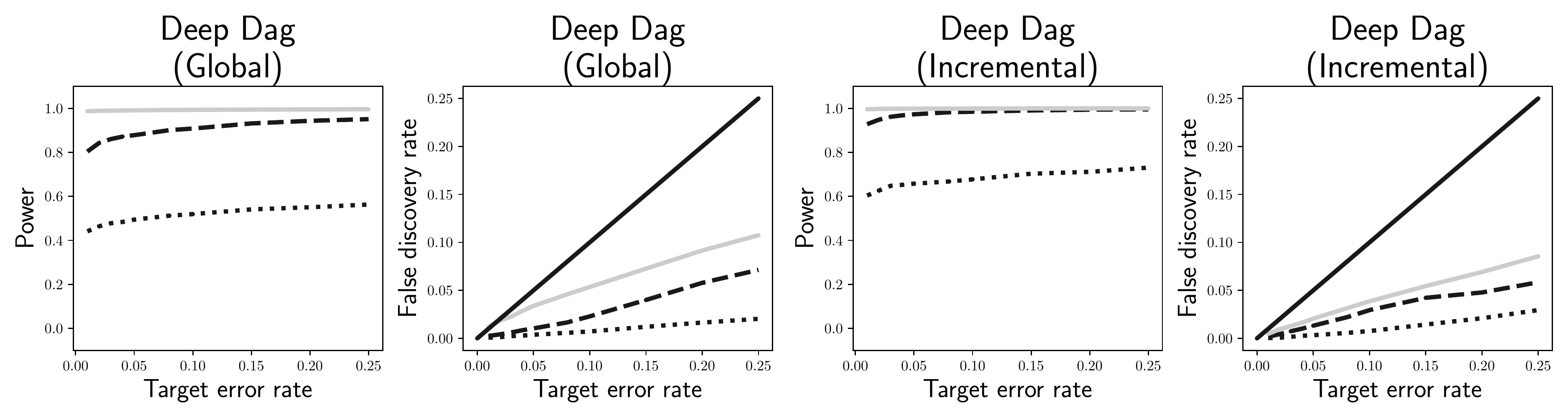}
  \caption{\label{fig:dependency-fdr} Power and empirical error rate for three different multiple hypothesis testing methods on two different simulations.  Gray lines indicate a Fisher smoothing version of DAGGER, dashed lines indicate a conservative Stouffer smoothing version of DAGGER, dotted lines indicate DAGGER, and the solid black line indicates the maximum allowable error.}
  \end{figure}

\begin{figure}
  \includegraphics[width=\linewidth]{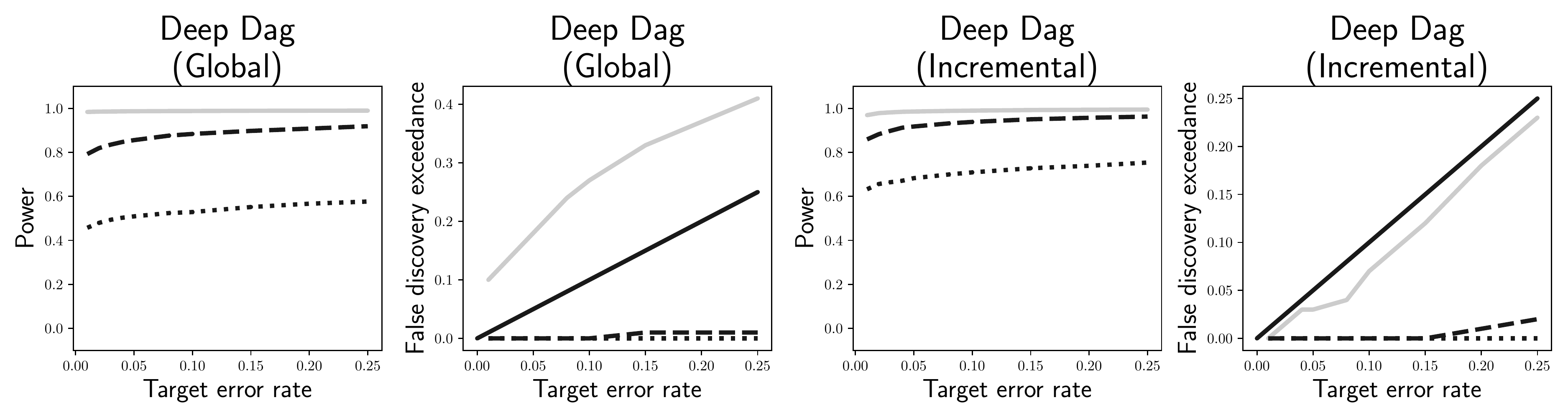}
  \caption{\label{fig:dependency-fdx} Power and empirical error rate for three different multiple hypothesis testing methods on two different simulations.  Gray lines indicate a Fisher smoothing version of the FDX extension of Meijer and Goeman, dashed lines indicate a conservative Stouffer smoothing version of the FDX extension of Meijer and Goeman, dotted lines indicate the FDX extension of Meijer and Goeman, and the solid black line indicates the maximum allowable error.}
  \end{figure}

Here we perform numerical experiments to investigate the performance of various methods in the presence of dependent null statistics with Gaussian copulas.  Several extant methods can provably control error in this scenario:

\begin{itemize}
  \item The estimators of Meijer and Goeman are still guaranteed to control familywise error.
  \item It follows that the FDX extension of \cite{meijer:goeman:2015:dag-fwer} (proposed in Algorithm \ref{alg:meijer}) still guarantees false discovery exceedance.  
  \item The DAGGER estimator is still guaranteed to control false discovery rate -- as long as the null statistics have nonnegative correlations. 
\end{itemize}

\input{meijer_algorithm}

Smoothing techniques can be used with each of these methods, yielding improved power in some cases.  We will consider two smoothing methods.
\begin{enumerate}
  \item The conservative Stouffer smoothing method:
  \[
  \psmooth_v \gets 
      \begin{cases}
          1 & \mathrm{if } \sum_{w \in \mathcal{C}_v} \pi_{vw} \Phi^{-1}(p_w)\geq 0 \\
          \Phi\left(\sum_{w \in \mathcal{C}_v} \pi_{vw} \Phi^{-1}(p_w)\right) & \mathrm{otherwise}.
      \end{cases}
  \]
  where $\Phi$ indicates the cumulative distribution function for a standard normal distribution.  In the experiments here, we focus on a specific choice for $\pi$, namely one which averages each node with its direct children:
  \begin{align*}
    \pi_{vw} &\propto \begin{cases}
      1 & \mathrm{if}\ v=w \\
      1 & \mathrm{if}\ v\rightarrow w \\
      0 & \mathrm{otherwise}
    \end{cases} \\
    \sum_w \pi_{vw} & = 1
  \end{align*}
  \item Fisher smoothing:
  \[
  \psmooth_v \gets 1-F_{\chi^2_{|2\mathcal{C}_v|}}\left(-2\sum_{w \in \mathcal{C}_v} \log p_w\right)
  \]
  where $F_{\chi^2_d}$ indicates the cumulative distribution function for a $\chi^2$ distribution with $d$ degrees of freedom. 
\end{enumerate}
Hypothesis testing can be performed by applying existing methods directly to the smoothed values.  As described in the main text and proved in \ref{subapp:generic}, the conservative Stouffer smoothing versions of several methods provably control type I errors:

\begin{itemize}
  \item The estimators of Meijer and Goeman applied to conservative Stouffer smoothed statistics are still guaranteed to control familywise error.
  \item The FDX extension of Meijer of Goeman applied to conservative Stouffer smoothed statistics still guarantees false discovery exceedance.  
  \item The DAGGER estimator applied to conservative Stouffer smoothed statistics is still guaranteed to control false discovery rate -- as long as the null statistics have nonnegative correlations.
\end{itemize}

To quantify the performance of these smoothing-based algorithms, we created simulated datasets in which the ground truth was known.  The simulated data was created in three stages:

\begin{enumerate}
  \item We first designed a five-layer directed acyclic graph structure with 50 nodes at each layer.  Each node has three randomly-selected parents from the layer above it.  
  \item We then selected which nodes would be considered nonnull and sampled a value for each of the nonnull nodes.  We performed this selection and sampling using two different methods:
    \begin{itemize}
        \item Global alternative. The value at each nonnull node is $\mathrm{Beta}(\exp(-4),0.5)$. \DAG{}s are populated starting at the leaves and null nodes are flipped to nonnull with probability $0.2$.  
        \item Incremental alternative. The value at each nonnull node is $\mathrm{Beta}(\exp(-4-0.3\times(D-d)),0.5)$, where $d$ is the depth of the node and $D$ is the maximum depth of the \DAG{}. The graph is populated starting at the leaves with nonnull probability $0.2$ and internal nodes are intersection hypotheses that are null if and only if all their child nodes are null.
    \end{itemize}
  \item Finally, we sampled values for the null nodes.  We constructed a Gaussian process on the graph structure:
  \[
  Z_v|Z_{\mathcal{A}_v} \sim \mathcal{N}\left(
    \frac{1}{|\mathcal{P}_v \cap \mathcal{\bar S}|} \sum_{w \in \mathcal{P}_v\cap \mathcal{\bar S}}Z_{w},
    1,
  \right)
  \]
  Here $\mathcal{P}_v$ denotes the direct parents of the node $v$ and $\mathcal{A}_v$ denotes all ancestors of $v$ in the graph.  We then computed marginally uniform values for each node; letting $F_v(c) = \mathrm{pr}(Z_v \leq c)$, we took $p_v = F_v(Z_v)$.
\end{enumerate}

For each type of alternative (global and incremental) we constructed 100 trials (yielding 200 trials in all).  We applied six different hypothesis testing methods to each trial and calculated the power and empirical error rates of each method.  The six methods are listed below.

\begin{itemize}
  \item DAGGER (which provably controls the false discovery rate in this setting).
  \item A Fisher smoothing version of DAGGER (which is not guaranteed to control the false discovery rate).
  \item A conservative Stouffer version of DAGGER (which provably controls the false discovery rate).
  \item The FDX extension of Meijer and Goeman (which provably controls the false discovery exceedance), tuned to limit the false discovery proportion below 10\%.
  \item A Fisher smoothing version of the FDX extension of Meijer and Goeman (which is not guaranteed to control false discovery exceedance), tuned to limit the false discovery proportion below 10\%.
  \item A conservative Stouffer version of the FDX extension of Meijer and Goeman (which provably controls the false discovery exceedance).
\end{itemize}

The results are shown in Figure \ref{fig:dependency-fdr} and Figure \ref{fig:dependency-fdx}.  In each simulation we find that that the smoothed versions of algorithms have higher power than the original unsmoothed versions.  However, as expected, the Fisher smoothing versions occasionally violate their target error rates.  For example, the Fisher smoothing version of the FDX extension of Meijer and Goeman has an empirical false discovery exceedance which is roughly double its target error rate when applied to simulations based on the global alternative scheme.

\subsection{Intuitions for achieving the greatest power with smoothing techniques}

\label{subapp:badexperiments}

There are some cases in which some smoothed versions of extant algorithms have less power than the corresponding unsmoothed versions.  Trouble will arise if the smoothed value for a nonnull hypothesis is heavily influenced by null $p$-values.  However, there are many kinds of smoothing, and in most cases we found that there is \emph{some} form of smoothing which yields higher power.  Ideally the user may use a-priori knowledge to choose a smoothing method wisely.  For example, if user suspects many of the null hypotheses to be correct, it may be unwise to include smoothing functions which are heavily influenced by the $p$-values for those hypotheses.  

We designed a numerical experiment to demonstrate the power of different smoothing methods in different contexts.  We hope this may help guide users in their thinking about what smoothing techniques may be appropriate for different kinds of data.  We created four simulated datasets.  Each was created in three stages:

\begin{enumerate}
  \item We first designed a five-layer directed acyclic graph structure with 50 nodes at each layer.  Each node has three randomly-selected parents from the layer above it.  
  \item We then selected which nodes would be considered nonnull:  for the first simulation only the nodes in the first layer were considered nonnull, for the the second simulation only the nodes in the first two layers were considered nonnull, and so-on.  
  \item We then sampled a value for each node.  Values at nonnull nodes were sampled according to $\mathrm{Beta}(0.1,0.5)$ and values at null nodes were sampled uniformly between 0 and 1.
\end{enumerate}

We then ran several different hypothesis selection algorithms on each of the four simulations.  We investigated three types of smoothing:
\begin{enumerate}
  \item No smoothing, i.e. with smoothing function $f_v(p)=p_v$
  \item Fisher smoothing using only direct children, i.e. with smoothing function
  \[
    f_v(p)=-2\log p_v -\sum_{w:\ v\in \mathcal{P}_w} 2\log p_w,
  \]
  where $\mathcal{P}_w$ indicates the parents of node $w$.
  \item Fisher smoothing using all children, i.e. with smoothing function
  \[
    f_v(p)=\sum_{c \in \mathcal{C}_v} 2\log p_c.
  \]
\end{enumerate}
We investigated two classes of hypothesis selection algorithms: the DAGGER method and the FDX extension of Meijer and Goeman.  In total, this yielded six hypothesis selection algorithms:
\begin{enumerate}
  \item the DAGGER method, 
  \item the FDX extension of Meijer and Goeman,
  \item a Fisher smoothing version of the DAGGER method which only uses the direct children of each node, 
  \item a Fisher smoothing version of the FDX extension of Meijer and Goeman which only uses the direct children of each node, 
  \item a Fisher smoothing version of the DAGGER method which uses all descendants of each node, and
  \item a Fisher smoothing version of the FDX extension of Meijer and Goeman which uses all descendants of each node.
\end{enumerate}
The DAGGER-based algorithms were tuned to limit the false discovery rate to at most 5\%.  The Meijer and Goeman algorithms were tuned to limit the false discovery proportion below 10\% with probability at least 95\%.

\begin{figure}
  \begin{minipage}[c]{.45\linewidth}
  \includegraphics[width=\linewidth]{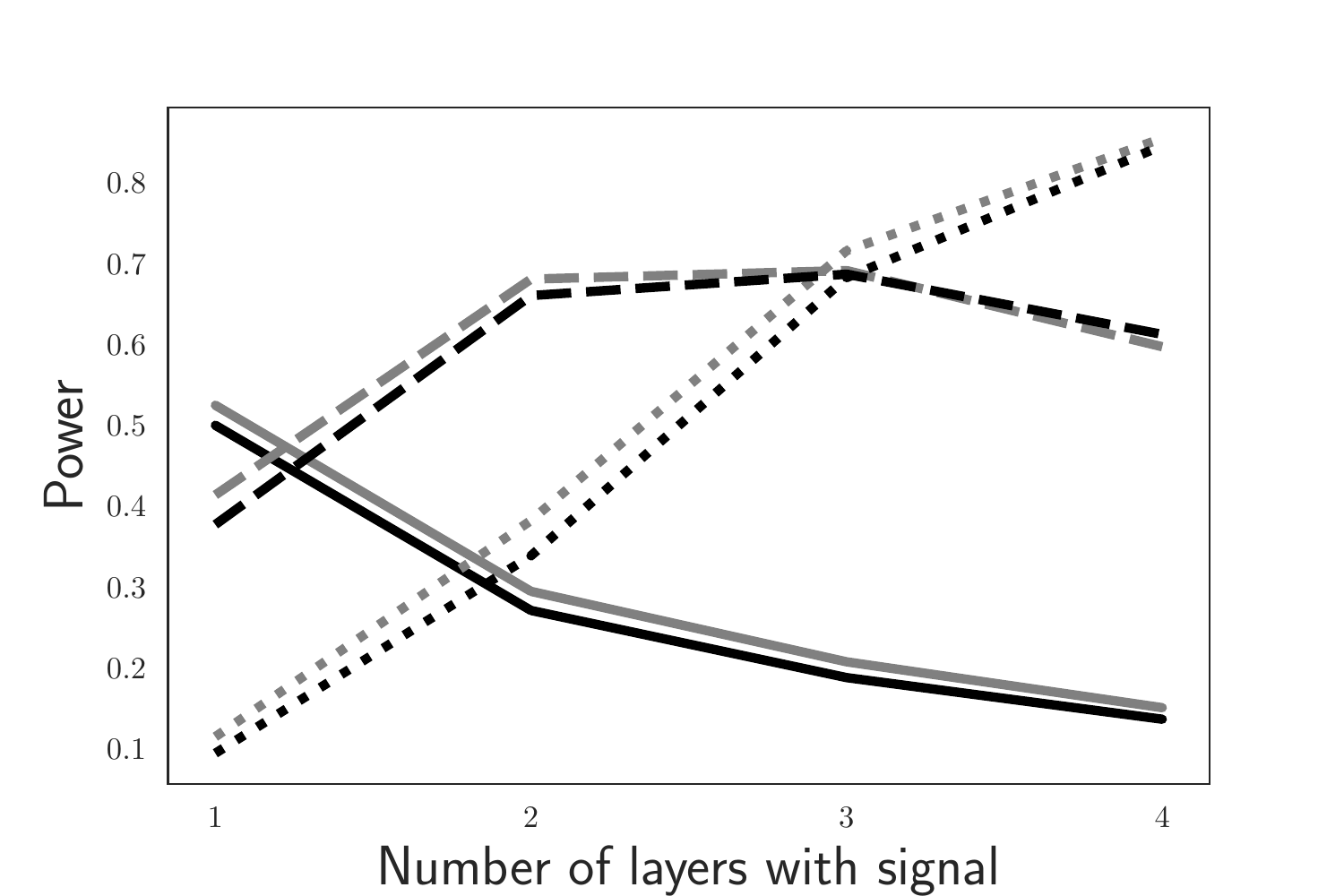}
  \end{minipage}\hfill
  \begin{minipage}[c]{.52\linewidth}
    \caption{\label{fig:badpic} To get the highest possible power, the hypothesis selection algorithm must be chosen differently for different simulation types.  Gray lines indicate algorithms based on Dagger and black lines indicate algorithms based on Meijer and Goeman.  Solid lines indicate the original algorithms, dashed lines indicate smoothed versions of those algorithms which use information from each node along with its direct descendants, and dotted lines indicate smoothed versions which use information from all descendants of each node.  In simulations where more layers include nonnull hypothesis, algorithms with more smoothing perform better.}
  \end{minipage}
  \end{figure} 

The results are shown in Figure \ref{fig:badpic}.  To get the highest possible power, the hypothesis selection algorithm must be chosen differently for different simulation types.  When the hypotheses in layers 2-5 are all correct, the best choices would be methods 1 and 2 from the list above (i.e. algorithms without any smoothing).  When only the hypotheses in layers 3-5 are correct, methods 3 and 4 are best (i.e. algorithms which involve only a local smoothing over the direct children of each node).  When only the hypotheses in layers 4-5 are correct, methods 5 and 6 are best (i.e. algorithms which smooth over all descendants of each node).

In summary, smoothing methods can help most when the smoothed values for nonnull hypotheses are most heavily influenced by $p$-values from other nonnull hypothesis.  It follows that the user should select smoothing functions such that this property holds as often as possible.  Unfortunately, this may be difficult to do because the user does not know a-priori which hypotheses are correct and which are not.  We caution that some users may feel tempted to choose smoothing techniques based on the data itself; we emphasize such an approach is completely unacceptable unless the user has sufficient data to create a clean split between data used to choose the hypothesis selection algorithm and data used to conduct the final hypothesis tests.  A deeper investigation of how users should make these decisions is merited, but beyond the scope of the present work.
 
\label{subapp:bad}


%% file: fdx_algorithm.tex
\begin{algorithm}[H]
\caption{\label{alg:fdx} Hybrid method to control \fdx{}}
\SetKwInOut{Input}{input}\SetKwInOut{Output}{output}
\Input{Smoothed $p$-values $\td{p} = (\td{p}_1, \ldots, \td{p}_n)$, directed acyclic graph $\mathcal{G}$, exceedance level $\gamma, \alpha$}
$\rejected_0 \leftarrow \text{Algorithm 1 with inputs } \td{p}, \mathcal{G}, \alpha$\;
$\rejected'\leftarrow \text{first }\lfloor |\rejected|\gamma / (1 - \gamma)\rfloor \text{ elements of }\mathrm{topological\_sort}(\mathcal{G}\setminus \rejected_0)$\;
\KwRet{$\rejected_0\cup \rejected'$}

\end{algorithm}

%% file: meijer_algorithm.tex
\begin{algorithm}[H]
\SetKwInOut{Input}{input}\SetKwInOut{Output}{output}
\Input{smoothed $p$-values $\td{p} = (\td{p}_1, \ldots, \td{p}_n)$, directed acyclic graph $\mathcal{G}$, target level $\alpha$}
$\rejected\gets \emptyset$\;
\Repeat{$\rejected$ does not change}{
  $(\pi_v)_{v\not\in \rejected}\gets$  Algorithm~\ref{alg:filling}\;
  $\rejected\gets \rejected\cup \{v\not\in \rejected: \texttt{parents}(v)\subset \rejected  \texttt{ and }p_v \le \alpha\pi_v\}$
}
\KwRet{$\rejected$}
\caption{\label{alg:meijer} Meijer-Goeman procedure}
\end{algorithm}

%% file: main.bbl
\begin{thebibliography}{54}
\providecommand{\natexlab}[1]{#1}
\providecommand{\url}[1]{\texttt{#1}}
\expandafter\ifx\csname urlstyle\endcsname\relax
  \providecommand{\doi}[1]{doi: #1}\else
  \providecommand{\doi}{doi: \begingroup \urlstyle{rm}\Url}\fi

\bibitem[Balabdaoui and Wellner(2014)]{balabdaoui2014chernoff}
F.~Balabdaoui and J.~A. Wellner.
\newblock Chernoff's density is log-concave.
\newblock \emph{Bernoulli}, 20\penalty0 (1):\penalty0 231, 2014.

\bibitem[Barber and Cand{\`e}s(2015)]{barber2015controlling}
R.~F. Barber and E.~J. Cand{\`e}s.
\newblock Controlling the false discovery rate via knockoffs.
\newblock \emph{Annals of Statistics}, 43\penalty0 (5):\penalty0 2055--2085,
  2015.

\bibitem[Barber and Ramdas(2017)]{barber:ramdas:2017:p-filter}
R.~F. Barber and A.~Ramdas.
\newblock The p-filter: {M}ultilayer false discovery rate control for grouped
  hypotheses.
\newblock \emph{Journal of the Royal Statistical Society: {{S}}eries {B}
  (Statistical Methodology)}, 2017.

\bibitem[Benjamini and Bogomolov(2014)]{benjamini2014selective}
Y.~Benjamini and M.~Bogomolov.
\newblock Selective inference on multiple families of hypotheses.
\newblock \emph{Journal of the Royal Statistical Society: {S}eries {B}
  (Statistical Methodology)}, pages 297--318, 2014.

\bibitem[Benjamini and Hochberg(1995)]{benjamini:hochberg:1995:bh}
Y.~Benjamini and Y.~Hochberg.
\newblock Controlling the false discovery rate: {{A}} practical and powerful
  approach to multiple testing.
\newblock \emph{Journal of the Royal Statistical Society: {{S}}eries {B}
  (Statistical Methodology)}, 57\penalty0 (1):\penalty0 289--300, 1995.

\bibitem[Benjamini and
  Yekutieli(2001)]{benjamini:yekutieli:2001:fdr-dependence}
Y.~Benjamini and D.~Yekutieli.
\newblock The control of the false discovery rate in multiple testing under
  dependency.
\newblock \emph{The Annals of Statistics}, 29\penalty0 (4):\penalty0
  1165--1188, 2001.

\bibitem[Block et~al.(1985)Block, Savits, and Shaked]{block1985concept}
H.~W. Block, T.~H. Savits, and M.~Shaked.
\newblock A concept of negative dependence using stochastic ordering.
\newblock \emph{Statistics \& Probability Letters}, 3\penalty0 (2):\penalty0
  81--86, 1985.

\bibitem[Block et~al.(1987)Block, Bueno, Savits, and
  Shaked]{block1987probability}
H.~W. Block, V.~Bueno, T.~H. Savits, and M.~Shaked.
\newblock Probability inequalities via negative dependence for random variables
  conditioned on order statistics.
\newblock \emph{Naval Research Logistics}, 34\penalty0 (4):\penalty0 547--554,
  1987.

\bibitem[Bogomolov et~al.(2017)Bogomolov, Peterson, Benjamini, and
  Sabatti]{bogomolov:etal:2017:simes-tree}
M.~Bogomolov, C.~B. Peterson, Y.~Benjamini, and C.~Sabatti.
\newblock Testing hypotheses on a tree: new error rates and controlling
  strategies.
\newblock \emph{arXiv preprint arXiv:1705.07529}, 2017.

\bibitem[Bonferroni(1936)]{bonferroni1936teoria}
C.~Bonferroni.
\newblock Teoria statistica delle classi e calcolo delle probabilita.
\newblock \emph{Pubblicazioni del R Istituto Superiore di Scienze Economiche e
  Commericiali di Firenze}, 8:\penalty0 3--62, 1936.

\bibitem[Brown(1975)]{brown1975400}
M.~B. Brown.
\newblock 400: {A} method for combining non-independent, one-sided tests of
  significance.
\newblock \emph{Biometrics}, pages 987--992, 1975.

\bibitem[Cormen et~al.(2009)Cormen, Leiserson, Rivest, and
  Stein]{cormen:2001:toposort}
T.~Cormen, C.~Leiserson, R.~Rivest, and C.~Stein.
\newblock \emph{Introduction To Algorithms}.
\newblock 2009.

\bibitem[Costanzo et~al.(2019)Costanzo, Kuzmin, van Leeuwen, Mair, Moffat,
  Boone, and Andrews]{costanzo:etal:2019:gene-interaction-maps-review}
M.~Costanzo, E.~Kuzmin, J.~van Leeuwen, B.~Mair, J.~Moffat, C.~Boone, and
  B.~Andrews.
\newblock Global genetic networks and the genotype-to-phenotype relationship.
\newblock \emph{Cell}, 177\penalty0 (1):\penalty0 85--100, 2019.

\bibitem[Dixit et~al.(2016)Dixit, Parnas, Li, Chen, Fulco, Jerby-Arnon,
  Marjanovic, Dionne, Burks, and Raychowdhury]{dixit:etal:2016:perturb-seq}
A.~Dixit, O.~Parnas, B.~Li, J.~Chen, C.~P. Fulco, L.~Jerby-Arnon, N.~D.
  Marjanovic, D.~Dionne, T.~Burks, and R.~Raychowdhury.
\newblock Perturb-{S}eq: {D}issecting molecular circuits with scalable
  single-cell {RNA} profiling of pooled genetic screens.
\newblock \emph{Cell}, 167\penalty0 (7):\penalty0 1853--1866, 2016.

\bibitem[Donoho et~al.(2004)Donoho, Jin, et~al.]{donoho2004higher}
D.~Donoho, J.~Jin, et~al.
\newblock Higher criticism for detecting sparse heterogeneous mixtures.
\newblock \emph{The Annals of Statistics}, 32\penalty0 (3):\penalty0 962--994,
  2004.

\bibitem[Efron(1965)]{efron1965increasing}
B.~Efron.
\newblock Increasing properties of p{\'o}lya frequency function.
\newblock \emph{The Annals of Mathematical Statistics}, pages 272--279, 1965.

\bibitem[Fisher(1925)]{fisher:stats-methods}
R.~A. Fisher.
\newblock \emph{Statistical methods for research workers}.
\newblock 1925.

\bibitem[Genovese and Wasserman(2006)]{genovese:wasserman:2006:fdx}
C.~R. Genovese and L.~Wasserman.
\newblock Exceedance control of the false discovery proportion.
\newblock \emph{Journal of the American Statistical Association}, 101\penalty0
  (476):\penalty0 1408--1417, 2006.

\bibitem[Goeman and Mansmann(2008)]{goeman2008multiple}
J.~J. Goeman and U.~Mansmann.
\newblock Multiple testing on the directed acyclic graph of gene ontology.
\newblock \emph{Bioinformatics}, 24\penalty0 (4):\penalty0 537--544, 2008.

\bibitem[Goeman and Solari(2010)]{goeman:solari:2010:sequentialrp}
J.~J. Goeman and A.~Solari.
\newblock The sequential rejection principle of familywise error control.
\newblock \emph{The Annals of Statistics}, pages 3782--3810, 2010.

\bibitem[Heard and Rubin-Delanchy(2018)]{heard:rubin-delanchy:2018:p-merging}
N.~A. Heard and P.~Rubin-Delanchy.
\newblock Choosing between methods of combining-values.
\newblock \emph{Biometrika}, 105\penalty0 (1):\penalty0 239--246, 2018.

\bibitem[Holm(1979)]{holm:1979:bonferroni-holm}
S.~Holm.
\newblock A simple sequentially rejective multiple test procedure.
\newblock \emph{Scandinavian Journal of Statistics}, pages 65--70, 1979.

\bibitem[Kamae et~al.(1977)Kamae, Krengel, and O'Brien]{kamae1977stochastic}
T.~Kamae, U.~Krengel, and G.~L. O'Brien.
\newblock Stochastic inequalities on partially ordered spaces.
\newblock \emph{The Annals of Probability}, pages 899--912, 1977.

\bibitem[Katsevich and
  Sabatti(2019)]{katsevich:sabatti:2019:multilayer-knockoffs}
E.~Katsevich and C.~Sabatti.
\newblock Multilayer knockoff filter: {C}ontrolled variable selection at
  multiple resolutions.
\newblock \emph{The Annals of Applied Statistics}, 13\penalty0 (1):\penalty0 1,
  2019.

\bibitem[Kost and McDermott(2002)]{kost2002combining}
J.~T. Kost and M.~P. McDermott.
\newblock Combining dependent p-values.
\newblock \emph{Statistics \& Probability Letters}, 60\penalty0 (2):\penalty0
  183--190, 2002.

\bibitem[Kuzmin et~al.(2018)Kuzmin, VanderSluis, Wang, Tan, Deshpande, Chen,
  Usaj, Balint, Usaj, and
  Van~Leeuwen]{kuzmin:etal:2018:yeast-trigenic-interactions}
E.~Kuzmin, B.~VanderSluis, W.~Wang, G.~Tan, R.~Deshpande, Y.~Chen, M.~Usaj,
  A.~Balint, M.~M. Usaj, and J.~Van~Leeuwen.
\newblock Systematic analysis of complex genetic interactions.
\newblock \emph{Science}, 360\penalty0 (6386), 2018.

\bibitem[Lei and Fithian(2016)]{lei2016power}
L.~Lei and W.~Fithian.
\newblock Power of ordered hypothesis testing.
\newblock In \emph{International Conference on Machine Learning}, pages
  2924--2932, 2016.

\bibitem[Lei and Fithian(2018)]{lei:fithian:2018:adapt}
L.~Lei and W.~Fithian.
\newblock {{AdaPT}}: {A}n interactive procedure for multiple testing with side
  information.
\newblock \emph{Journal of the Royal Statistical Society: {S}eries {B}
  (Statistical Methodology)}, 80\penalty0 (4):\penalty0 649--679, 2018.

\bibitem[Lei et~al.(2017)Lei, Ramdas, and Fithian]{lei:etal:2017:star}
L.~Lei, A.~Ramdas, and W.~Fithian.
\newblock {STAR}: {{A}} general interactive framework for {{FDR}} control under
  structural constraints.
\newblock \emph{arXiv preprint arXiv:1710.02776}, 2017.

\bibitem[Li and Barber(2017)]{li:barber:2017:accumulation_tests}
A.~Li and R.~F. Barber.
\newblock Accumulation tests for {{FDR}} control in ordered hypothesis testing.
\newblock \emph{Journal of the American Statistical Association}, 112\penalty0
  (518):\penalty0 837--849, 2017.

\bibitem[Li and Barber(2019)]{li:barber:2019:sabha}
A.~Li and R.~F. Barber.
\newblock Multiple testing with the structure-adaptive {B}enjamini--{H}ochberg
  algorithm.
\newblock \emph{Journal of the Royal Statistical Society: {S}eries {B}
  (Statistical Methodology)}, 81\penalty0 (1):\penalty0 45--74, 2019.

\bibitem[Littell and Folks(1971)]{littell1971asymptotic}
R.~C. Littell and J.~L. Folks.
\newblock Asymptotic optimality of fisher's method of combining independent
  tests.
\newblock \emph{Journal of the American Statistical Association}, 66\penalty0
  (336):\penalty0 802--806, 1971.

\bibitem[Liu and Xie(2020)]{liu2020cauchy}
Y.~Liu and J.~Xie.
\newblock {C}auchy combination test: {A} powerful test with analytic p-value
  calculation under arbitrary dependency structures.
\newblock \emph{Journal of the American Statistical Association}, 115\penalty0
  (529):\penalty0 393--402, 2020.

\bibitem[Lynch(2014)]{lynch:2014:dissertation}
G.~Lynch.
\newblock \emph{The control of the false discovery rate under structured
  hypotheses}.
\newblock PhD thesis, New Jersey Institute of Technology, 2014.

\bibitem[Lynch and Guo(2016)]{lynch2016procedures}
G.~Lynch and W.~Guo.
\newblock On procedures controlling the {FDR} for testing hierarchically
  ordered hypotheses.
\newblock \emph{arXiv preprint arXiv:1612.04467}, 2016.

\bibitem[Marcus et~al.(1976)Marcus, Eric, and Gabriel]{marcus1976closed}
R.~Marcus, P.~Eric, and K.~R. Gabriel.
\newblock On closed testing procedures with special reference to ordered
  analysis of variance.
\newblock \emph{Biometrika}, 63\penalty0 (3):\penalty0 655--660, 1976.

\bibitem[Meijer and Goeman(2015)]{meijer:goeman:2015:dag-fwer}
R.~J. Meijer and J.~J. Goeman.
\newblock A multiple testing method for hypotheses structured in a directed
  acyclic graph.
\newblock \emph{Biometrical Journal}, 57\penalty0 (1):\penalty0 123--143, 2015.

\bibitem[Meinshausen(2008)]{meinshausen:2008:tree-fwer}
N.~Meinshausen.
\newblock Hierarchical testing of variable importance.
\newblock \emph{Biometrika}, 95\penalty0 (2):\penalty0 265--278, 2008.

\bibitem[Ramdas et~al.(2019{\natexlab{a}})Ramdas, Chen, Wainwright, and
  Jordan]{ramdas:etal:2019:dagger}
A.~Ramdas, J.~Chen, M.~J. Wainwright, and M.~I. Jordan.
\newblock A sequential algorithm for false discovery rate control on directed
  acyclic graphs.
\newblock \emph{Biometrika}, 106\penalty0 (1):\penalty0 69--86,
  2019{\natexlab{a}}.

\bibitem[Ramdas et~al.(2019{\natexlab{b}})Ramdas, Barber, Wainwright, and
  Jordan]{ramdas:etal:2019:p-filter}
A.~K. Ramdas, R.~F. Barber, M.~J. Wainwright, and M.~I. Jordan.
\newblock A unified treatment of multiple testing with prior knowledge using
  the p-filter.
\newblock \emph{The Annals of Statistics}, 47\penalty0 (5):\penalty0
  2790--2821, 2019{\natexlab{b}}.

\bibitem[Rosenbaum(2008)]{rosenbaum2008testing}
P.~R. Rosenbaum.
\newblock Testing hypotheses in order.
\newblock \emph{Biometrika}, 95\penalty0 (1):\penalty0 248--252, 2008.

\bibitem[R{\"u}ger(1978)]{ruger1978maximale}
B.~R{\"u}ger.
\newblock Das maximale signifikanzniveau des tests: {L}ehneh o ab, wennk untern
  gegebenen tests zur ablehnung f{\"u}hren.
\newblock \emph{Metrika}, 25\penalty0 (1):\penalty0 171--178, 1978.

\bibitem[Scott et~al.(2015)Scott, Kelly, Smith, Zhou, and
  Kass]{scott:etal:2015:fdr-regression}
J.~G. Scott, R.~C. Kelly, M.~A. Smith, P.~Zhou, and R.~E. Kass.
\newblock False discovery rate regression: {A}n application to neural synchrony
  detection in primary visual cortex.
\newblock \emph{Journal of the American Statistical Association}, 110\penalty0
  (510):\penalty0 459--471, 2015.

\bibitem[Shaffer(1995)]{shaffer1995multiple}
J.~P. Shaffer.
\newblock Multiple hypothesis testing.
\newblock \emph{Annual Review of Psychology}, 46\penalty0 (1):\penalty0
  561--584, 1995.

\bibitem[Simes(1986)]{simes1986improved}
R.~J. Simes.
\newblock An improved {B}onferroni procedure for multiple tests of
  significance.
\newblock \emph{Biometrika}, 73\penalty0 (3):\penalty0 751--754, 1986.

\bibitem[Stouffer et~al.(1949)Stouffer, Suchman, DeVinney, Star, and
  Williams~Jr]{stouffer:etal:1949:merging}
S.~A. Stouffer, E.~A. Suchman, L.~C. DeVinney, S.~A. Star, and R.~M.
  Williams~Jr.
\newblock \emph{The {A}merican soldier: {A}djustment during army life},
  volume~1.
\newblock 1949.

\bibitem[Tansey et~al.(2018)Tansey, Wang, Blei, and
  Rabadan]{tansey:etal:2018:bb-fdr}
W.~Tansey, Y.~Wang, D.~Blei, and R.~Rabadan.
\newblock Black box {FDR}.
\newblock In \emph{International Conference on Machine Learning}, pages
  4874--4883, 2018.

\bibitem[Tippett(1931)]{tippett1931methods}
L.~H.~C. Tippett.
\newblock \emph{The methods of statistics: {A}n introduction mainly for workers
  in the biological sciences}.
\newblock 1931.

\bibitem[Vesely et~al.(2021)Vesely, Finos, and Goeman]{vesely2021permutation}
A.~Vesely, L.~Finos, and J.~J. Goeman.
\newblock Permutation-based true discovery guarantee by sum tests.
\newblock \emph{arXiv preprint arXiv:2102.11759}, 2021.

\bibitem[Vovk and Wang(2020)]{vovk2020combining}
V.~Vovk and R.~Wang.
\newblock Combining p-values via averaging.
\newblock \emph{Biometrika}, 107\penalty0 (4):\penalty0 791--808, 2020.

\bibitem[Vovk et~al.(2020)Vovk, Wang, and Wang]{vovk2020admissible}
V.~Vovk, B.~Wang, and R.~Wang.
\newblock Admissible ways of merging p-values under arbitrary dependence.
\newblock \emph{arXiv preprint arXiv:2007.14208}, 2020.

\bibitem[Wang et~al.(2014)Wang, Wei, Sabatini, and
  Lander]{wang:etal:2014:crispr-knockout-screens}
T.~Wang, J.~J. Wei, D.~M. Sabatini, and E.~S. Lander.
\newblock Genetic screens in human cells using the {CRISPR}-{Cas9} system.
\newblock \emph{Science}, 343\penalty0 (6166):\penalty0 80--84, 2014.

\bibitem[Xia et~al.(2017)Xia, Zhang, Zou, and Tse]{xia:etal:2017:neuralfdr}
F.~Xia, M.~J. Zhang, J.~Y. Zou, and D.~Tse.
\newblock {NeuralFDR}: {L}earning discovery thresholds from hypothesis
  features.
\newblock In \emph{Advances in Neural Information Processing Systems}, pages
  1541--1550, 2017.

\bibitem[Yekutieli(2008)]{yekutieli2008hierarchical}
D.~Yekutieli.
\newblock Hierarchical false discovery rate--controlling methodology.
\newblock \emph{Journal of the American Statistical Association}, 103\penalty0
  (481):\penalty0 309--316, 2008.

\end{thebibliography}
